\documentclass[11pt]{article}

\pdfoutput=1 

\usepackage{subfig}
\usepackage{tikz}
\usetikzlibrary{arrows,backgrounds,calc,fit}

\usepackage{multirow} 

\newcounter{gadgetCounter}
\newcommand{\gadgetSubfloat}[2][]{\subfloat[#1]{#2}\refstepcounter{gadgetCounter}}

\newcounter{basicGadgetCounter}
\newcommand{\basicGadgetSubfloat}[2][]{\subfloat[#1]{#2}\refstepcounter{basicGadgetCounter}}

\usepackage{tyson}

\newcommand{\holant}[2]{\operatorname{Holant}(#1\hspace{0.5mm}|\hspace{0.5mm}#2)}
\newcommand{\Holant}{\operatorname{Holant}}
\newcommand{\GR}{(\mathcal{G}\hspace{0.5mm}|\hspace{0.5mm}\mathcal{R})}

\def\borderColor{blue!60}
\def\arrowType{stealth}
\def\scale{0.6}
\def\nodeDist{1.4cm}
\def\startPos{0.8}
\def\endPos{0.2}
\tikzstyle{internal} = [draw,fill,shape=circle]
\tikzstyle{external} = [shape=circle]

\newcommand{\shrinkMatrixTwoRows}[1]{\mbox{\scriptsize$#1$}}
\newcommand{\shrinkMatrixFourRows}[1]{\mbox{\scriptsize$#1$}}

\title{Gadgets and Anti-Gadgets Leading to a Complexity Dichotomy}

\author{
 Jin-Yi Cai\\
 \footnotesize University of Wisconsin-Madison\\
 \footnotesize \texttt{jyc@cs.wisc.edu}
 \and
 Michael Kowalczyk\\
 \footnotesize Northern Michigan University\\
 \footnotesize \texttt{mkowalcz@nmu.edu}
 \and
 Tyson Williams\\
 \footnotesize University of Wisconsin-Madison\\
 \footnotesize \texttt{tdw@cs.wisc.edu}
}

\date{}

\begin{document}
\maketitle
\setcounter{page}{0}
\thispagestyle{empty}


\begin{abstract}
We introduce an idea called anti-gadgets in complexity reductions. These combinatorial gadgets have the effect of erasing the presence of some other graph fragment, as if we had managed to include a negative copy of a graph gadget. We use this idea to prove a complexity dichotomy theorem for the partition function $Z(G)$ on 3-regular directed graphs $G$, where each edge is given a complex-valued binary function $f: \{0,1\}^2 \rightarrow \mathbb{C}$. We show that \[Z(G) = \sum_{\sigma: V(G) \to \{0,1\}} \prod_{(u,v) \in E(G)} f(\sigma(u), \sigma(v)),\] is either computable in polynomial time or \#P-hard, depending explicitly on $f$.

To state the dichotomy theorem more explicitly, we show that the partition function $Z(G)$ on 3-regular directed graphs $G$ is computable in polynomial time when $f$ belongs to one of four classes, which can be described as (1) degenerate, (2) generalized disequality, (3) generalized equality, and (4) affine after a holographic transformation. In all other cases it is \#P-hard. Here class (4), after a holographic transformation, can also be described as an exponential quadratic polynomial of the form $i^{Q(x,y)}$, where $i = \sqrt{-1}$ and the cross term $xy$ in the quadratic polynomial $Q(x,y)$ has an even coefficient. If the input graph $G$ is planar, then an additional class of functions becomes computable in polynomial time, and everything else remains \#P-hard. This additional class is precisely those which can be computed by holographic algorithms with matchgates, making use of the Fisher-Kasteleyn-Temperley algorithm via Pfaffians.

There is a long history in the study of ``Exactly Solved Models'' in statistical physics. In the language of complexity theory, physicists' notion of an ``Exactly Solvable'' system corresponds to a system with a polynomial time computable partition function. A central question is to identify which ``systems'' can be solved ``exactly'' and which ``systems'' are ``difficult''. While in physics, there is no rigorous definition of being ``difficult'', complexity theory supplies the proper notion---\#P-hardness.

The main innovation in this paper is the idea of an anti-gadget. It is analogous to the pairing of a particle and its anti-particle in physics. Coupled with the idea of anti-gadgets, we also introduce a general way of proving \#P-hardness by two types of gadgets called recursive gadgets and projector gadgets. We prove a Group Lemma which spells out a general condition for the technique to succeed. This Group Lemma states that as long as the group generated by the transition matrices of the constructed gadgets is infinite, then one can interpolate all unary functions---a key step in the proof of \#P-hardness. Interpolation is carried out by forming a Vandermonde system and proving that it is of full rank. The anti-gadget concept makes the transition to group theory very natural and seamless.

Not only is the idea of anti-gadgets useful in proving a new complexity dichotomy theorem in counting complexity, we also show that anti-gadgets provide a simple explanation for some miraculous cancellations that were observed in previous results. Furthermore, anti-gadgets can also guide the search for gadget sets more by design than by chance.
\end{abstract}
\nocite{Kas61, Kas67}

\pagebreak

\section{Introduction}
Reduction, the method of transforming one problem to another, and thereby proving the hardness of a problem for an entire complexity class, is arguably the most successful tool in complexity theory to date. When expressed in terms of graph problems, a typical reduction from problem $\Pi_1$ to problem $\Pi_2$ is carried out by designing a \emph{gadget}---a graph fragment with some desirable properties. The reduction starts from an instance graph $G_1$ for $\Pi_1$ and introduces one or more copies of the gadget to obtain an instance graph $G_2$ (or possibly multiple instance graphs) for $\Pi_2$.

The graph $G_2$ may contain a polynomial number of copies of the gadget. \emph{But can it include some negative copies of a gadget?} Of course not; the notion of a \emph{negative} graph fragment seems meaningless. However, in this paper we introduce an idea in reduction theory that has \emph{the effect of} introducing \emph{negative copies} of a gadget in a reduction. More precisely, we show that our new construction idea, when expressed in algebraic terms, has the same effect as \emph{erasing} the presence of some graph fragment. It is \emph{as if} we managed to include a negative copy of a certain gadget. We call this an \emph{anti-gadget}. It is analogous to the pairing of a particle and its anti-particle in physics. We demonstrate the elegance and usefulness of anti-gadgets by proving a new complexity dichotomy theorem in counting complexity where anti-gadgets play a decisive role. Furthermore, we show that anti-gadgets provide a simple explanation for some miraculous cancellations that were observed in previous results~\cite{CK10, CK11}. We also observe how anti-gadgets can guide the search for such gadget sets more by design than by chance.

The new dichotomy theorem that we prove using anti-gadgets can be stated in terms of spin systems on 3-regular graphs with vertices taking values in $\{0,1\}$ and an arbitrary complex-valued edge function $f(\cdot, \cdot)$ that is not necessarily symmetric. Define the \emph{partition function} on $G = (V, E)$ as $Z(G) = \sum_{\sigma: V(G) \to \{0,1\}} \prod_{(u,v) \in E(G)} f(\sigma(u), \sigma(v))$. Depending on the nature of the edge function $f$, we show that the problem $Z(\cdot)$ is either tractable in $\P$ or $\SHARPP$-hard. More precisely, the problem is $\SHARPP$-hard unless the edge function is (\textit{i}) degenerate, (\textit{ii}) generalized equality, (\textit{iii}) generalized disequality, or is (\textit{iv}) affine after a holographic transformation. For these four classes of functions, the problem is computable in polynomial time. Furthermore, if the input is restricted to planar graphs, then the class of tractable problems is augmented by those which are solvable by holographic algorithms with matchgates---all other problems remain $\SHARPP$-hard. Thus, holographic algorithms with matchgates are a \emph{universal} methodology for this class of counting problems over directed 3-regular graphs, which are $\SHARPP$-hard in general, but become tractable on planar graphs.

The main innovation in this paper is the idea of an anti-gadget. In terms of concrete theorems proved, this paper can be viewed as extending previous dichotomy theorems for the complexity of the spin system for {\em symmetric} edge functions~\cite{CLX08, CLX09b, KC10, CK10, CK11} to {\em asymmetric} edge functions. The new dichotomy theorem holds over 3-regular graphs, for {\em any} (not necessarily symmetric) complex-valued edge function. In physics, the 0-1 vertex assignments are called spins, and the edge function values $f(\sigma(u), \sigma(v))$ correspond to local interactions between particles. There is a long history in the statistical physics community in the study of ``Exactly Solved Models''~\cite{Bax82, TF61}. In the language of modern complexity theory, physicists' notion of an ``Exactly Solvable'' system corresponds to a system with polynomial time computable partition function. A central question is to identify which ``systems'' can be solved ``exactly'' and which ``systems'' are ``difficult''. While in physics, there is no rigorous definition of being ``difficult'', complexity theory supplies the proper notion---$\SHARPP$-hardness.

The class of problems we study in this paper has a close connection with holant problems~\cite{XZZ07, Val08, CLX09a, CLX10, CHL10, MikeThesis, CL11, GLV11, CLX11}. We use holographic algorithms~\cite{Val08, CL11} to prove both tractability and $\SHARPP$-hardness. In general, holant problems are a natural class of counting problems which can encode all counting Constraint Satisfaction Problems (\#CSP)~\cite{CKS01} and graph homomorphisms. Dichotomy theorems for graph homomorphisms~\cite{Lov67, BG05, CC10, DGP07, DG00, GGJT10, CCL09, HN90} and \#CSP~\cite{Bul06, Bul08, BD07, BG05, CHL10, CCL11, DR10, DGJ09, GHLX11, CK11} have been a very active research area. Compared to \#CSP and graph homomorphisms, the main difficulty here is bounded degree, which makes hardness proofs more challenging, and for a good reason---there are indeed more tractable cases.

\section{Notation and Background}
The partition function on directed graphs is a special case of Holant problems defined as follows. A \emph{signature grid} $\Omega = (G, \mathcal{F}, \pi)$ consists of a labeled undirected graph $G = (V,E)$ where $\pi$ labels each vertex $v \in V$ with a function $f_v \in \mathcal{F}$. The inputs of $f_v$ are identified with the incident edges $E(v)$ at $v$. For any edge assignment $\xi: E \rightarrow \{0,1\}$, $f_v(\xi \mid_{E(v)})$ is the evaluation, and the counting problem is to compute $\Holant_{\Omega} = \sum_{\xi : E \to \{0,1\}} \prod_{v \in V} f_v(\xi \mid_{E(v)})$.

Given any directed 3-regular graph $G = (V, E)$, its edge-vertex incidence graph $G'$ has vertex set $V(G') = V \union E$ and edge set $E(G') = \{(v, e) \mid \text{$v$ is incident to $e$ in $G$}\}$. The graph $G'$ is bipartite and $(2,3)$-regular. If we label each $v \in V \subset V(G')$ with the \textsc{Equality} function $=_3$ of arity 3 and each $e \in E \subset V(G')$ with the original edge function $f$ from $G$, then the Holant value on $G'$ is exactly the partition function $Z(G)$. Essentially $=_3$ forces all incident edges in $G'$ at a vertex $v \in V \subset V(G')$ to take the same value, which reduces to vertex assignments on $V$, as in $Z(G)$. We frequently take this bipartite perspective of $Z(G)$ as holant problems in order to use holographic transformations, which is more convenient on bipartite graphs.

A function $f: \{0,1\}^k \rightarrow \C$ can be denoted by $(f_0, f_1, \ldots, f_{2^k-1})$, where $f_i$ is the value of $f$ on the $i$th lexicographical bit string of length $k$. They are also called \emph{signatures}. A signature $f$ of arity $k$ is degenerate if $f$ is a tensor product of unary signatures: $f = (a_1, b_1) \tensor \cdots \tensor (a_k, b_k)$. For $(2,3)$-regular bipartite graphs $(U, V, E)$, if every $u \in U$ is labeled $f$ and every $v \in V$ is labeled $r$, then we also use $\holant{f}{r}$ to denote the holant problem. Our main result is a dichotomy theorem for $\holant{f}{{=}_3}$, for an arbitrary binary function $f = (w,x,y,z)$, where $w,x,y,z \in \C$. It has the same complexity as $\holant{c f}{{=}_3}$ for any nonzero $c \in \C$, hence we often normalize a signature by a nonzero scalar. More generally, if $\mathcal{G}$ and $\mathcal{R}$ are finite sets of signatures and the vertices of~$U$ (resp.~$V$) are labeled by signatures from $\mathcal{G}$ (resp.~$\mathcal{R}$), then we also use $\holant{\mathcal{G}}{\mathcal{R}}$ to denote the bipartite holant problem. Signatures in $\mathcal{G}$ are called \emph{generators} and signatures in $\mathcal{R}$ are called \emph{recognizers}.

Signatures from $\mathcal{G}$ and $\mathcal{R}$ are available at each vertex of the appropriate part of an input graph. Instead of a single vertex, we can use graph fragments to generalize this notion. A $\GR$-gate $\Gamma$ is a triple $(H, \mathcal{G}, \mathcal{R})$, where $H = (U,V,E,D)$ is a bipartite graph with some dangling edges $D$. Other than these dangling edges, a $\GR$-gate is the same as a signature grid. The purpose of dangling edges is to provide input and output edges. In $H = (U,V,E,D)$, each node in $U$ (resp.~$V$) is assigned a function in $\mathcal{G}$ (resp.~$\mathcal{R}$), $E$ are the regular edges, and $D$ are the dangling edges. The $\GR$-gate $\Gamma$ defines a function: $\Gamma(y_1, y_2, \ldots, y_q) = \sum_{(x_1, x_2, \ldots, x_p) \in \{0,1\}^p} H(x_1, x_2, \ldots, x_p, y_1, y_2, \ldots, y_q)$, where $p=|E|$, $q=|D|$, $(y_1, y_2, \ldots, y_q) \in \{0,1\}^q$ denotes an assignment on the dangling edges, and $H(x_1, x_2, \ldots, x_p, y_1, y_2, \ldots, y_q)$ denotes the product of evaluations at every vertex of $H$. We also call this function the signature of the $\GR$-gate $\Gamma$. A $\GR$-gate can be used in a signature grid as if it is just a single node with the same signature. Signature grids on bipartite $(2,3)$-regular graphs can be identified with directed 3-regular graphs, where we merge two incident edges at every vertex $w$ of degree 2, and label the new edge by the arity 2 signature $f_w$. The edge is oriented from $u$ to $v$ if $f_w = f_w(u, v)$. Figure~\ref{fig:gadget:bipartiteVsDirected} gives an example of a $\GR$-gate both as a bipartite (2,3)-regular graph and as an equivalent directed 3-regular graph.

\begin{figure}[t]
 \def\vsWidth{6cm}
 \def\vsScale{0.45}
 \def\vsNodeDist{1cm}
 \def\vsStartPos{0.9}
 \def\vsEndPos{0.1}
 \centering
 \captionsetup[subfigure]{width=\vsWidth}
 \subfloat[As a bipartite (2,3)-regular graph]{
  \makebox[\vsWidth][c]{
  \begin{tikzpicture}[scale=\vsScale,transform shape,>=\arrowType,node distance=\vsNodeDist,semithick]
   \node[external] (0)              {};
   \node[internal] (2) [right of=0] {};
   \node[internal] (4) [right of=2] {};
   \node[external] (6) [right of=4] {};
   \node[internal] (8) [below of=4] {};
   \node[internal] (5) [below of=8] {};
   \node[external] (7) [right of=5] {};
   \node[internal] (3) [left  of=5] {};
   \node[external] (1) [left  of=3] {};
   \path (0) edge [-] node[pos=\vsStartPos] (e1) {} (2)
         (1) edge [-] node[pos=\vsStartPos] (e2) {} (3)
         (2) edge [-]                               (4)
         (3) edge [-]                               (5)
         (4) edge [-] node[pos=\vsEndPos]   (e3) {} (6)
             edge [-]                               (8)
         (5) edge [-] node[pos=\vsEndPos]   (e4) {} (7)
             edge [-]                               (8);
   \begin{pgfonlayer}{background}
    \node[draw=\borderColor,thick,rounded corners,fit = (2) (3) (4) (5) (8) (e1) (e2) (e3) (e4)] {};
   \end{pgfonlayer}
  \end{tikzpicture}}}
 \quad
 \subfloat[As a directed 3-regular graph]{
  \makebox[\vsWidth][c]{
  \begin{tikzpicture}[scale=\vsScale,transform shape,>=\arrowType,node distance=\vsNodeDist,semithick]
   \node[external] (0)              {};
   \node[external] (2) [right of=0] {};
   \node[internal] (4) [right of=2] {};
   \node[external] (6) [right of=4] {};
   \node[external] (8) [below of=4] {};
   \node[internal] (5) [below of=8] {};
   \node[external] (7) [right of=5] {};
   \node[external] (3) [left  of=5] {};
   \node[external] (1) [left  of=3] {};
   \path (0) edge [->] node[pos=\vsStartPos] (e1) {} (4)
         (1) edge [->] node[pos=\vsStartPos] (e2) {} (5)
         (4) edge [->]                               (5)
             edge [-]  node[pos=\vsEndPos]   (e3) {} (6)
         (5) edge [-]  node[pos=\vsEndPos]   (e4) {} (7);
   \begin{pgfonlayer}{background}
    \node[draw=\borderColor,thick,rounded corners,fit = (2) (3) (4) (5) (8) (e1) (e2) (e3) (e4)] {};
   \end{pgfonlayer}
  \end{tikzpicture}}}
 \caption{The two representations of a $\GR$-gate}
 \label{fig:gadget:bipartiteVsDirected}
\end{figure}
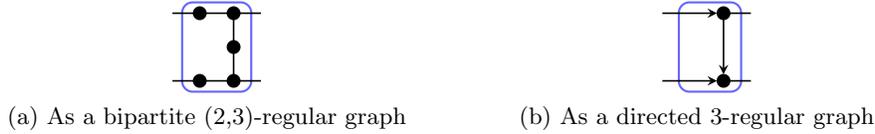

We designate dangling edges as either \emph{leading} edges or \emph{trailing} edges. Each $\GR$-gate is pictured with leading edges protruding to the left and any trailing edges to the right. Suppose a $\GR$-gate has $m$ leading edges and $n$ trailing edges. Then the signature of the $\GR$-gate can be organized as a $2^m$-by-$2^n$ \emph{transition matrix} $M$, where the row (resp.~column) is indexed by a $\{0,1\}$-assignment to the leading (resp.~trailing) edges. When pictured, if $e_1, \ldots, e_n$ are $n$ dangling edges in top-down order, then $b_1 \ldots b_n \in \{0,1\}^n$ is the index for the assignment where $e_i$ is assigned $b_i$. We denote the transition matrix of Gadget $i$ as $M_i$ unless otherwise noted.

The constructions in this paper are primarily based upon two kinds of $\GR$-gates, which we call \emph{recursive gadgets} and \emph{projector gadgets}. An \emph{arity-$d$ recursive $\GR$-gadget} is a $\GR$-gate with $d$ leading edges and $d$ trailing edges. A $\GR$-gate is a \emph{projector $\GR$-gadget from arity $n$ to $m$} if it has $m$ leading edges and $n$ trailing edges. Internally, for both recursive and projector gadgets, we require that all leading edges connect to a degree 2 vertex (equivalently a directed edge), while all trailing edges connect to a degree 3 vertex. These gadget types are defined in this way to maintain the bipartite structure of the signature grid when we merge trailing edges of one gadget with leading edges of another (see Figure~\ref{fig:gadget:interpolationExamples}).

\begin{figure}[t]
 \def\exWidth{3.2cm}
 \def\exScale{0.45}
 \def\exNodeDist{1cm}
 \centering
 \subfloat[Projector gadget]{
  \makebox[\exWidth][c]{
  \begin{tikzpicture}[scale=\exScale,transform shape,>=\arrowType,node distance=\nodeDist,semithick]
   \node[external]  (0)              {};
   \node[internal]  (1) [right of=0] {};
   \node[internal]  (2) [above right of=1] {};
   \node[internal]  (3) [below right of=1] {};
   \node[external]  (4) [right of=2] {};
   \node[external]  (5) [right of=3] {};
   \node[external]  (6) [above of=4] {};
   \node[external]  (7) [below of=5] {};
   \node[external]  (8) [above of=2,yshift=-5] {};
   \node[external]  (9) [below of=3,yshift=5]  {};
   \path (0) edge [->] node[pos=\startPos] (e1) {} (1)
         (1) edge [->]                             (2)
             edge [->]                             (3)
         (2) edge [-]  node[pos=\endPos]   (e2) {} (4)
             edge [-]  node[pos=\endPos]   (e3) {} (6)
         (3) edge [-]  node[pos=\endPos]   (e4) {} (5)
             edge [-]  node[pos=\endPos]   (e5) {} (7);
   \begin{pgfonlayer}{background}
    \node[draw=\borderColor,thick,rounded corners,fit = (1) (2) (3) (8) (9) (e1) (e2) (e3) (e4) (e5)] {};
   \end{pgfonlayer}
  \end{tikzpicture}}}
 \subfloat[Recursive gadget]{
  \makebox[\exWidth][c]{
  \begin{tikzpicture}[scale=\exScale,transform shape,>=\arrowType,node distance=\nodeDist,semithick]
   \node[external]  (0)              {};
   \node[external]  (1) [below of=0] {};
   \node[internal]  (2) [right of=1] {};
   \node[internal]  (3) [right of=0] {};
   \node[external]  (4) [right of=2] {};
   \node[external]  (5) [right of=3] {};
   \node[external] (n1) [above of=3,yshift=-35] {};
   \node[external] (n2) [below of=2,yshift=35]  {};
   \path (0) edge [->] node[pos=\startPos] (e1) {} (3)
         (1) edge [->] node[pos=\startPos] (e2) {} (2)
         (2) edge [<-]                             (3)
             edge [-]  node[pos=\endPos]   (e3) {} (4)
         (3) edge [-]  node[pos=\endPos]   (e4) {} (5);
   \begin{pgfonlayer}{background}
    \node[draw=\borderColor,thick,rounded corners,fit = (n1) (n2) (e1) (e2) (e3) (e4)] {};
   \end{pgfonlayer}
  \end{tikzpicture}}}
 \subfloat[Planar embedding of interpolation construction]{
  \begin{tikzpicture}[scale=\exScale,transform shape,>=\arrowType,node distance=\nodeDist,semithick]
   \node[external]   (0)               {};
   \node[internal]   (1) [right of =0] {};
   \node[internal]   (2) [above right of=1] {};
   \node[internal]   (3) [below right of=1] {};
   \node[external]   (5) [right of =2] {};
   \node[external]   (4) [above of =5] {};
   \node[external]   (6) [right of =3] {};
   \node[external]   (7) [below of =6] {};
   \node[internal]   (8) [right of =6] {};
   \node[internal]   (9) [right of =7] {};
   \node[internal]  (10) [right of =8] {};
   \node[internal]  (11) [right of =9] {};
   \node[external]  (12) [right of=10] {};
   \node[external]  (13) [right of=11] {};
   \node[external]  (14) [right of=12] {};
   \node[external]  (15) [right of=13] {};
   \node[internal]  (16) [right of=14] {};
   \node[internal]  (17) [right of=15] {};
   \path let
          \p1 = (12),
          \p2 = (15)
         in
          node[external] (n1) at (\x1 / 2 + \x2 / 2, \y1 / 2 + \y2 / 2) {\Huge \dots};
   \node[external]  (18) [right of=16] {};
   \node[external]  (19) [right of=17] {};
   \node[external]  (20) [right of=18] {};
   \node[external]  (21) [right of=19] {};
   \node[external]  (23) [right of= 5] {};
   \node[external]  (24) [right of=23] {};
   \node[external]  (25) [right of=24] {};
   \node[external]  (26) [right of=25] {};
   \node[external]  (27) [right of=26] {};
   \node[external]  (28) [right of=27] {};
   \node[external]  (29) [right of=28] {};
   \node[external]  (30) [above of=29] {};
   \path let
          \p1 = (18.west),
          \p2 = (21)
         in
          node[external] (n2) at (\x1 / 2 + \x2 / 2, \y1 / 2 + \y2 / 2) {\Huge \dots};
   \path  (0) edge [->] node[pos=\startPos] (e1) {}  (1)
          (1) edge [->]                              (2)
              edge [->]                              (3)
          (2) edge [<-]                              (4.center)
              edge [<-]                             (29)
          (3) edge [->]                              (8)
              edge [-]                               (7.center)
   (4.center) edge [-]                              (30)
   (7.center) edge [->]                              (9)
          (8) edge [->]                              (9)
              edge [->]                             (10)
          (9) edge [->]                             (11)
         (10) edge [->]                             (11)
              edge [-]                              (12)
         (11) edge [-]                              (13)
         (14) edge [->]                             (16)
         (15) edge [->]                             (17)
         (16) edge [->]                             (17)
              edge [-]                              (18)
         (17) edge [-]                              (19)
         (20) edge [-, out=0, in=0]                 (29)
              edge [-, out=180, in=0, looseness=0]  (20)
         (21) edge [-, out=180, in=0, looseness=0]  (21)
              edge [-, out=0, in=0] coordinate (c1) (30)
         (29) edge [-, out=0, in=180, looseness=0]  (29)
         (30) edge [-, out=0, in=180, looseness=0]  (30);
   \begin{pgfonlayer}{background}
    \node[draw=\borderColor,thick,rounded corners,fit = (1) (21) (30) (e1) (c1)] {};
   \end{pgfonlayer}
  \end{tikzpicture} \label{gadget:example:interpolationConstruction}}
 \caption{Arity 4 to 1 projector and recursive gadgets and the construction for interpolation}
 \label{fig:gadget:interpolationExamples}
\end{figure}
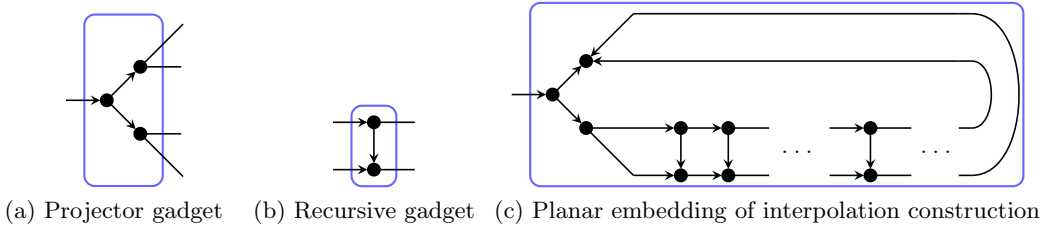

\section{Gadgets and Anti-Gadgets} \label{sec:gadgetsAndAnti-Gadgets}
In this section, we start with a gentle primer to the association between a combinatorial gadget and its signature written as a transition matrix. We show that one can typically express the transition matrix starting from a few of the most basic gadget components and their matrices as atomic building blocks, after applying some well defined operations. We then introduce anti-gadgets and explain why they are so effective.

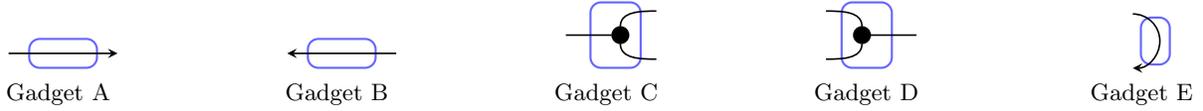
\begin{figure}[t]
 \centering
 \captionsetup[subfigure]{labelformat=empty}
 \basicGadgetSubfloat[Gadget \ref{gadget:simple:eRight}]{
  \begin{tikzpicture}[scale=\scale,transform shape,>=\arrowType,node distance=\nodeDist,semithick]
   \node[external] (0)              {};
   \node[external] (1) [right of=0] {};
   \node[external] (2) [right of=1] {};
   \path (0) edge [->] node[pos=0.3] (e1) {} node[pos=0.7] (e2) {} (2);
   \begin{pgfonlayer}{background}
    \node[draw=\borderColor,thick,rounded corners,fit = (1) (e1) (e2)] {};
   \end{pgfonlayer}
  \end{tikzpicture}} \label{gadget:simple:eRight}
 \hfill
 \basicGadgetSubfloat[Gadget \ref{gadget:simple:eLeft}]{
  \begin{tikzpicture}[scale=\scale,transform shape,>=\arrowType,node distance=\nodeDist,semithick]
   \node[external] (0)              {};
   \node[external] (1) [right of=0] {};
   \node[external] (2) [right of=1] {};
   \path (0) edge [<-] node[pos=0.3] (e1) {} node[pos=0.7] (e2) {} (2);
   \begin{pgfonlayer}{background}
    \node[draw=\borderColor,thick,rounded corners,fit = (1) (e1) (e2)] {};
   \end{pgfonlayer}
  \end{tikzpicture}} \label{gadget:simple:eLeft}
 \hfill
 \basicGadgetSubfloat[Gadget \ref{gadget:simple:projector}]{
  \begin{tikzpicture}[scale=\scale,transform shape,>=\arrowType,node distance=\nodeDist,semithick]
   \node[external] (0)              {};
   \node[internal] (1) [right of=0] {};
   \node[external] (2) [above right of=1,yshift=-13] {};
   \node[external] (3) [below right of=1,yshift=13]  {};
   \path (0) edge [-]                  node[pos=\startPos] (e1) {} (1)
         (1) edge [-, out=90, in=180]  node[pos=0.4]       (e2) {} (2)
             edge [-, out=-90, in=180] node[pos=0.4]       (e3) {} (3);
   \begin{pgfonlayer}{background}
    \node[draw=\borderColor,thick,rounded corners,fit = (1) (e1) (e2) (e3)] {};
   \end{pgfonlayer}
  \end{tikzpicture}} \label{gadget:simple:projector}
 \hfill
 \basicGadgetSubfloat[Gadget \ref{gadget:simple:expander}]{
  \begin{tikzpicture}[scale=\scale,transform shape,>=\arrowType,node distance=\nodeDist,semithick]
   \node[external] (0)             {};
   \node[internal] (1) [left of=0] {};
   \node[external] (2) [above left of=1,yshift=-13] {};
   \node[external] (3) [below left of=1,yshift=13]  {};
   \path (0) edge [-]                node[pos=\startPos] (e1) {} (1)
         (1) edge [-, out=90,  in=0] node[pos=0.4]       (e2) {} (2)
             edge [-, out=-90, in=0] node[pos=0.4]       (e3) {} (3);
   \begin{pgfonlayer}{background}
    \node[draw=\borderColor,thick,rounded corners,fit = (1) (e1) (e2) (e3)] {};
   \end{pgfonlayer}
  \end{tikzpicture}} \label{gadget:simple:expander}
 \hfill
 \def\exWidth{1.8cm}
 \captionsetup[subfigure]{width=\exWidth}
 \basicGadgetSubfloat[Gadget \ref{gadget:simple:linearStarter}]{
  \makebox[\exWidth][c]{
  \begin{tikzpicture}[scale=\scale,transform shape,>=\arrowType,node distance=\nodeDist,semithick]
   \node[external] (0) at (0.5cm,-0.4cm) {};
   \node[external] (1) at (0.5cm,-0.8cm) {};
  \draw[>=stealth,->] (0,0) arc (90:-90:0.6cm);
   \begin{pgfonlayer}{background}
    \node[draw=\borderColor,thick,rounded corners,fit = (0) (1)] {};
   \end{pgfonlayer}
  \end{tikzpicture}}} \label{gadget:simple:linearStarter}
 \caption{Five basic gadget components}
 \label{gadget:simple}
\end{figure}

We start with five basic gadget components as depicted in Figure~\ref{gadget:simple}. Their signature matrices are $\ref{gadget:simple:eRight} = \shrinkMatrixTwoRows{\begin{bmatrix} w & x\\ y & z \end{bmatrix}}$, $\ref{gadget:simple:eLeft} = \shrinkMatrixTwoRows{\begin{bmatrix} w & y\\ x & z \end{bmatrix}}$, $\ref{gadget:simple:projector} = \shrinkMatrixTwoRows{\begin{bmatrix} 1 & 0 & 0 & 0\\ 0 & 0 & 0 & 1 \end{bmatrix}}$, $\ref{gadget:simple:expander} = \shrinkMatrixFourRows{\begin{bmatrix} 1 & 0\\ 0 & 0\\ 0 & 0\\ 0 & 1 \end{bmatrix}}$, and $\ref{gadget:simple:linearStarter} = \shrinkMatrixFourRows{\begin{bmatrix} w\\ x\\ y\\ z \end{bmatrix}}$.

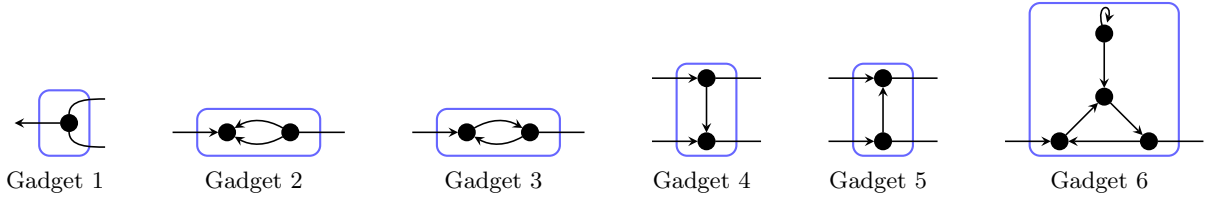
\begin{figure}[t]
 \centering
 \captionsetup[subfigure]{labelformat=empty}
 \gadgetSubfloat[Gadget \ref{gadget:main:finish:0:0}]{
  \begin{tikzpicture}[scale=\scale,transform shape,>=\arrowType,node distance=\nodeDist,semithick]
   \node[external] (0)              {};
   \node[internal] (1) [right of=0] {};
   \node[external] (2) [above right of=1,yshift=-13] {};
   \node[external] (3) [below right of=1,yshift=13]  {};
   \path (0) edge [<-]                 node[pos=\startPos] (e1) {} (1)
         (1) edge [-, out=90, in=180]  node[pos=0.4]       (e2) {} (2)
             edge [-, out=-90, in=180] node[pos=0.4]       (e3) {} (3);
   \begin{pgfonlayer}{background}
    \node[draw=\borderColor,thick,rounded corners,fit = (1) (e1) (e2) (e3)] {};
   \end{pgfonlayer}
  \end{tikzpicture}} \label{gadget:main:finish:0:0}
 \hfill
 \gadgetSubfloat[Gadget \ref{gadget:main:unary:0:111}]{
  \begin{tikzpicture}[scale=\scale,transform shape,>=\arrowType,node distance=\nodeDist,semithick]
   \node[external] (0)              {};
   \node[internal] (2) [right of=0] {};
   \node[internal] (1) [right of=2] {};
   \node[external] (3) [right of=1] {};
   \path (0) edge [->]             node[pos=\startPos] (e1) {} (2)
         (1) edge [->, bend right] node[midway]        (e2) {} (2)
             edge [->, bend left]  node[midway]        (e3) {} (2)
             edge [-]              node[pos=\endPos]   (e4) {} (3);
   \begin{pgfonlayer}{background}
    \node[draw=\borderColor,thick,rounded corners,fit = (1) (2) (e1) (e2) (e3) (e4)] {};
   \end{pgfonlayer}
  \end{tikzpicture}} \label{gadget:main:unary:0:111}
 \hfill
 \gadgetSubfloat[Gadget \ref{gadget:main:unary:0:110}]{
  \begin{tikzpicture}[scale=\scale,transform shape,>=\arrowType,node distance=\nodeDist,semithick]
   \node[external] (0)              {};
   \node[internal] (2) [right of=0] {};
   \node[internal] (1) [right of=2] {};
   \node[external] (3) [right of=1] {};
   \path (0) edge [->]             node[pos=\startPos] (e1) {} (2)
         (1) edge [<-, bend right] node[midway]        (e2) {} (2)
             edge [->, bend left]  node[midway]        (e3) {} (2)
             edge [-]              node[pos=\endPos]   (e4) {} (3);
   \begin{pgfonlayer}{background}
    \node[draw=\borderColor,thick,rounded corners,fit = (1) (2) (e1) (e2) (e3) (e4)] {};
   \end{pgfonlayer}
  \end{tikzpicture}} \label{gadget:main:unary:0:110}
 \hfill
 \gadgetSubfloat[Gadget \ref{gadget:main:binary:0:110}]{
  \begin{tikzpicture}[scale=\scale,transform shape,>=\arrowType,node distance=\nodeDist,semithick]
   \node[external] (0)              {};
   \node[external] (1) [below of=0] {};
   \node[internal] (2) [right of=1] {};
   \node[internal] (3) [right of=0] {};
   \node[external] (4) [right of=2] {};
   \node[external] (5) [right of=3] {};
   \path (0) edge [->] node[pos=\startPos] (e1) {} (3)
         (1) edge [->] node[pos=\startPos] (e2) {} (2)
         (2) edge [<-]                             (3)
             edge [-]  node[pos=\endPos]   (e3) {} (4)
         (3) edge [-]  node[pos=\endPos]   (e4) {} (5);
   \begin{pgfonlayer}{background}
    \node[draw=\borderColor,thick,rounded corners,fit = (2) (3) (e1) (e2) (e3) (e4)] {};
   \end{pgfonlayer}
  \end{tikzpicture}} \label{gadget:main:binary:0:110}
 \hfill
 \gadgetSubfloat[Gadget \ref{gadget:main:binary:0:111}]{
  \begin{tikzpicture}[scale=\scale,transform shape,>=\arrowType,node distance=\nodeDist,semithick]
   \node[external] (0)              {};
   \node[external] (1) [below of=0] {};
   \node[internal] (2) [right of=1] {};
   \node[internal] (3) [right of=0] {};
   \node[external] (4) [right of=2] {};
   \node[external] (5) [right of=3] {};
   \path (0) edge [->] node[pos=\startPos] (e1) {} (3)
         (1) edge [->] node[pos=\startPos] (e2) {} (2)
         (2) edge [->]                             (3)
             edge [-]  node[pos=\endPos]   (e3) {} (4)
         (3) edge [-]  node[pos=\endPos]   (e4) {} (5);
   \begin{pgfonlayer}{background}
    \node[draw=\borderColor,thick,rounded corners,fit = (2) (3) (e1) (e2) (e3) (e4)] {};
   \end{pgfonlayer}
  \end{tikzpicture}} \label{gadget:main:binary:0:111}
 \hfill
 \gadgetSubfloat[Gadget \ref{gadget:main:unary:4:101010}]{
  \begin{tikzpicture}[scale=\scale,transform shape,>=\arrowType,node distance=\nodeDist,semithick]
   \node[external] (0)                    {};
   \node[internal] (4) [right of=0]       {};
   \node[internal] (3) [above right of=4] {};
   \node[internal] (1) [below right of=3] {};
   \node[internal] (2) [above of=3]       {};
   \node[external] (5) [right of=1]       {};
   \path (0) edge [->]             node[pos=\startPos] (e1) {} (4)
         (1) edge [<-]                                         (3)
             edge [->]                                         (4)
         (2) edge [<-, loop above] coordinate          (c1)    (2)
             edge [->]                                         (3)
         (3) edge [<-]                                         (4)
         (1) edge [-]              node[pos=\endPos]   (e2) {} (5);
   \begin{pgfonlayer}{background}
    \node[draw=\borderColor,thick,rounded corners,fit = (1) (2) (3) (4) (e1) (e2) (c1)] {};
   \end{pgfonlayer}
  \end{tikzpicture}} \label{gadget:main:unary:4:101010}
 \caption{Recursive and projector gadgets}
\end{figure}


The first operation is matrix product, which corresponds to sequentially connecting two gadgets together. For example, Gadget~\ref{gadget:main:finish:0:0} is a simple composition of Gadget~\ref{gadget:simple:eLeft} and Gadget~\ref{gadget:simple:projector}, and thus its transition matrix is the matrix product $\ref{gadget:simple:eLeft} \ref{gadget:simple:projector} = \shrinkMatrixTwoRows{\begin{bmatrix} w & 0 & 0 & y\\ x & 0 & 0 & z \end{bmatrix}}$ (see Figure~\ref{gadget:decomposition:finish:0:0}). The second operation is tensor product, which corresponds to putting two gadgets in parallel (two disconnected parts). The transition matrix of Gadget~\ref{gadget:main:unary:0:111} is $\ref{gadget:simple:eRight} \ref{gadget:simple:projector} \ref{gadget:simple:eLeft}^{\tensor 2} \ref{gadget:simple:expander} = \shrinkMatrixTwoRows{\begin{bmatrix} w^3 + x^3 & w y^2 + x z^2\\ w^2 y + x^2 z & y^3 + z^3 \end{bmatrix}}$, where $\ref{gadget:simple:eLeft}^{\tensor 2}$ corresponds to the parallel part of the gadget and is clearly visible in Figure~\ref{gadget:decomposition:unary:0:111}. Similarly, Gadget~\ref{gadget:main:unary:0:110} has signature matrix $\ref{gadget:simple:eRight} \ref{gadget:simple:projector} (\ref{gadget:simple:eRight} \tensor \ref{gadget:simple:eLeft}) \ref{gadget:simple:expander}$. Note that the order of the tensor product is to make the top leading edge for the row (respectively, the top trailing edge for the column) the most significant bit. The transition matrices of Gadgets~\ref{gadget:main:binary:0:110} and~\ref{gadget:main:binary:0:111} are respectively $\shrinkMatrixTwoRows{\begin{bmatrix} w & x\\ y & z \end{bmatrix}^{\tensor 2}} \diag(w,x,y,z)$ and $\shrinkMatrixTwoRows{\begin{bmatrix} w & x\\ y & z \end{bmatrix}^{\tensor 2}} \diag(w,y,x,z)$ and can be mechanically derived by our gadgetry calculus as $\ref{gadget:simple:eRight}^{\tensor 2} (\ref{gadget:simple:projector} \tensor I_2) (I_2 \tensor \ref{gadget:simple:eRight} \tensor I_2) (I_2 \tensor \ref{gadget:simple:expander})$ and $\ref{gadget:simple:eRight}^{\tensor 2} (\ref{gadget:simple:projector} \tensor I_2) (I_2 \tensor \ref{gadget:simple:eLeft} \tensor I_2) (I_2 \tensor \ref{gadget:simple:expander})$. The composition of Gadget~\ref{gadget:main:binary:0:110} is illustrated in Figure~\ref{gadget:decomposition:binary:0:110}. Gadget $\ref{gadget:simple:linearStarter}$ is used to create a self-loop, as in Gadget~\ref{gadget:main:unary:4:101010}, which has transition matrix $\ref{gadget:simple:eRight} \ref{gadget:simple:projector} (\ref{gadget:simple:eRight} \tensor I_2) (\ref{gadget:simple:projector} \tensor I_2) (\ref{gadget:simple:eLeft} \ref{gadget:simple:projector} \ref{gadget:simple:linearStarter} \tensor I_4) (\ref{gadget:simple:eRight} \tensor \ref{gadget:simple:eLeft}) \ref{gadget:simple:expander}$. A composition is given in Figure~\ref{fig:gadget:decomposition:unary:4:101010} of the appendix.

\begin{figure}[b]
 \def\decompWidth{4cm}
 \def\decompScale{0.45}
 \def\dottedNodeDist{0.8cm}
 \centering
 \captionsetup[subfigure]{width=\decompWidth}
 \subfloat[Composition of Gadget \ref{gadget:main:finish:0:0}]{
  \makebox[\decompWidth][c]{
  \begin{tikzpicture}[scale=\decompScale,transform shape,>=\arrowType,node distance=\nodeDist,semithick]
   \node[external] (0)              {};
   \node[external] (1) [right of=0] {};
   \node[external] (2) [right of=1] {};
   \node[external] (3) [right of=2,node distance=\dottedNodeDist] {};
   \node[internal] (4) [right of=3] {};
   \node[external] (5) [above right of=4,yshift=-13] {};
   \node[external] (6) [below right of=4,yshift=13]  {};
   \path (0) edge [<-]                    node[pos=0.3] (e1) {} node[pos=0.7] (e2) {} (2)
         (2) edge [-, dotted, very thick]                                             (3)
         (3) edge [-]                     node[pos=\startPos] (e3) {}                 (4)
         (4) edge [-, out=90, in=180]     node[pos=0.4]       (e4) {}                 (5)
             edge [-, out=-90, in=180]    node[pos=0.4]       (e5) {}                 (6);
   \begin{pgfonlayer}{background}
    \node[draw=\borderColor,thick,rounded corners,fit = (1) (e1) (e2)] {};
    \node[draw=\borderColor,thick,rounded corners,fit = (4) (e3) (e4) (e5)] {};
   \end{pgfonlayer}
  \end{tikzpicture} \label{gadget:decomposition:finish:0:0}}}
 \hfill
 \subfloat[Composition of Gadget \ref{gadget:main:unary:0:111}]{
  \begin{tikzpicture}[scale=\decompScale,transform shape,>=\arrowType,node distance=\nodeDist,semithick]
   \node[external]  (0)               {};
   \node[external]  (1) [right of= 0] {};
   \node[external]  (2) [right of= 1] {};
   \node[external]  (3) [right of= 2,node distance=\dottedNodeDist] {};
   \node[internal]  (4) [right of= 3] {};
   \node[external]  (5) [above right of=4,yshift=-13] {};
   \node[external]  (6) [below right of=4,yshift=13]  {};
   \node[external]  (7) [right of= 5,node distance=\dottedNodeDist] {};
   \node[external]  (8) [right of= 6,node distance=\dottedNodeDist] {};
   \node[external]  (9) [right of= 7] {};
   \node[external] (10) [right of= 8] {};
   \node[external] (11) [right of= 9] {};
   \node[external] (12) [right of=10] {};
   \node[external] (13) [right of=11,node distance=\dottedNodeDist] {};
   \node[external] (14) [right of=12,node distance=\dottedNodeDist] {};
   \node[internal] (15) [below right of=13,yshift=13] {};
   \node[external] (16) [right of=15] {};
   \path (0) edge [->]                    node[pos=0.3] (e1) {} node[pos=0.7] (e2) {}  (2)
         (2) edge [-, dotted, very thick]                                              (3)
         (3) edge [-]                     node[pos=\startPos] (e3) {}                  (4)
         (4) edge [-, out=90, in=180]     node[pos=0.4]       (e4) {}                  (5)
             edge [-, out=-90, in=180]    node[pos=0.4]       (e5) {}                  (6)
         (5) edge [-, dotted, very thick]                                              (7)
         (6) edge [-, dotted, very thick]                                              (8)
        (11) edge [->]                    node[pos=0.3] (e6) {} node[pos=0.7] (e7) {}  (7)
        (12) edge [->]                    node[pos=0.3] (e8) {} node[pos=0.7] (e9) {}  (8)
        (11) edge [-, dotted, very thick]                                             (13)
        (12) edge [-, dotted, very thick]                                             (14)
        (15) edge [-, out=90,  in=0]      node[pos=0.4] (e10) {}                      (13)
             edge [-, out=-90, in=0]      node[pos=0.4] (e11) {}                      (14)
             edge [-]                     node[pos=\endPos] (e12) {}                  (16);
   \begin{pgfonlayer}{background}
    \node[draw=\borderColor,thick,rounded corners,fit = (1) (e1) (e2)] {};
    \node[draw=\borderColor,thick,rounded corners,fit = (4) (e3) (e4) (e5)] {};
    \node[draw=\borderColor,thick,rounded corners,fit = (9) (10) (e6) (e7) (e8) (e9)] {};
    \node[draw=\borderColor,thick,rounded corners,fit = (15) (e10) (e11) (e12)] {};
   \end{pgfonlayer}
  \end{tikzpicture} \label{gadget:decomposition:unary:0:111}}
 \hfill
 \subfloat[Composition of Gadget \ref{gadget:main:binary:0:110}]{
  \begin{tikzpicture}[scale=\decompScale,transform shape,>=\arrowType,node distance=\nodeDist,semithick]
   \node[external]  (0)               {};
   \node[external]  (2) [right of= 0] {};
   \node[external]  (4) [right of= 2] {};
   \node[external]  (6) [right of= 4,node distance=\dottedNodeDist] {};
   \node[internal]  (8) [right of= 6] {};
   \node[external]  (9) [above right of=8,yshift=-13] {};
   \node[external] (10) [below right of=8,yshift=13]  {};
   \node[external] (12) [right of= 9,node distance=\dottedNodeDist] {};
   \node[external] (13) [right of=10,node distance=\dottedNodeDist] {};
   \node[external] (15) [right of=12] {};
   \node[external] (16) [right of=13] {};
   \node[external] (18) [right of=15] {};
   \node[external] (19) [right of=16] {};
   \node[external] (21) [right of=18,node distance=\dottedNodeDist] {};
   \node[external] (22) [right of=19,node distance=\dottedNodeDist] {};
   \node[internal] (24) [below right of=22,yshift=13] {};
   \node[external] (23) [below left  of=24,yshift=13] {};
   \node[external] (20) [left  of=23,node distance=\dottedNodeDist] {};
   \node[external] (17) [left  of=20] {};
   \node[external] (14) [left  of=17] {};
   \node[external] (11) [left  of=14,node distance=\dottedNodeDist] {};
   \node[external] (26) [right of=24] {};
   \path let
          \p1 = (6),
          \p2 = (11)
         in
          node[external] (7) at (\x1,\y2) {};
   \node[external]  (5) [left  of= 7,node distance=\dottedNodeDist] {};
   \node[external]  (3) [left  of= 5] {};
   \node[external]  (1) [left  of= 3] {};
   \path let
          \p1 = (26),
          \p2 = (21)
         in
          node[external] (25) at (\x1,\y2) {};
   \path (0) edge [->]                    node[pos=0.3] (e1) {} node[pos=0.7] (e2) {}    (4)
         (1) edge [->]                    node[pos=0.3] (e3) {} node[pos=0.7] (e4) {}    (5)
         (4) edge [-, dotted, very thick]                                                (6)
         (5) edge [-, dotted, very thick]                                                (7)
         (6) edge [-]                     node[pos=\startPos] (e5) {}                    (8)
         (7) edge [-]                     node[midway]        (e6) {}                   (11)
         (8) edge [-, out=90, in=180]     node[pos=0.4]       (e7) {}                    (9)
             edge [-, out=-90, in=180]    node[pos=0.4]       (e8) {}                   (10)
         (9) edge [-, dotted, very thick]                                               (12)
        (10) edge [-, dotted, very thick]                                               (13)
        (11) edge [-, dotted, very thick]                                               (14)
        (12) edge [-]                     node[pos=0.3]  (e9) {} node[pos=0.7] (e10) {} (18)
        (13) edge [->]                    node[pos=0.3] (e11) {} node[pos=0.7] (e12) {} (19)
        (14) edge [-]                     node[pos=0.3] (e13) {} node[pos=0.7] (e14) {} (20)
        (18) edge [-, dotted, very thick]                                               (21)
        (19) edge [-, dotted, very thick]                                               (22)
        (20) edge [-, dotted, very thick]                                               (23)
        (24) edge [-, out=90, in=0]       node[pos=0.4]     (e13) {}                    (22)
             edge [-, out=-90, in=0]      node[pos=0.4]     (e14) {}                    (23)
        (21) edge [-]                     node[midway]      (e15) {}                    (25)
        (24) edge [-]                     node[pos=\endPos] (e16) {}                    (26);
   \begin{pgfonlayer}{background}
    \node[draw=\borderColor,thick,rounded corners,fit = (2) (3) (e1) (e2) (e3) (e4)] {};
    \node[draw=\borderColor,thick,rounded corners,fit = (8) (e5) (e6) (e7) (e8)] {};
    \node[draw=\borderColor,thick,rounded corners,fit = (15) (16) (17) (e9) (e10) (e11) (e12)] {};
    \node[draw=\borderColor,thick,rounded corners,fit = (24) (e13) (e14) (e15) (e16)] {};
   \end{pgfonlayer}
  \end{tikzpicture} \label{gadget:decomposition:binary:0:110}}
 \caption{Gadget compositions using the basic gadget components in Figure \ref{gadget:simple}}
\end{figure}
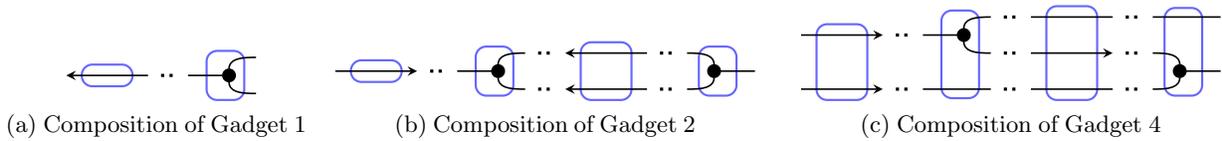

Now we introduce a powerful new technique called \emph{anti-gadgets}.

\begin{definition}
 Let $G$ be a recursive gadget with transition matrix $M$. Then a recursive gadget $G'$ is called an \emph{anti-gadget} of $G$ if the transition matrix of $G'$ is $\lambda M^{-1}$, for some $\lambda \in \C - \{0\}$.
\end{definition}

A crucial ingredient in our proof of $\SHARPP$-hardness is to produce an arbitrarily large set of pairwise linearly independent signatures. These signatures are used to form a Vandermonde system of full rank. One common way to produce an arbitrarily large set of signatures is to compose copies of a recursive gadget. Let $M$ be the transition matrix of some recursive gadget $G$. As discussed above, composing $k$ copies of $G$ produces a gadget with transition matrix $M^k$. If $M$ has infinite order (up to a scalar), then we have an arbitrarily large set of pairwise linearly independent signatures. Now suppose that $M$ has finite order (up to a scalar), that is, for some positive integer $k$, $M^k = \lambda I$, a nonzero multiple of the identity matrix. Then composing only $k-1$ copies of $G$ results in a gadget with a transition matrix that is the \emph{inverse} of $G$'s transition matrix (up to a scalar). This \emph{is} an anti-gadget of $G$.

If an anti-gadget of $G$ is composed with another gadget containing similar structure to that of $G$, then cancellations ensue and the composition yields a transition matrix that can be quite easy to analyze. E.g., Gadgets~\ref{gadget:main:binary:0:110} and~\ref{gadget:main:binary:0:111} only differ by the orientation of the vertical edge. When composing an anti-gadget of Gadget~\ref{gadget:main:binary:0:110} with Gadget~\ref{gadget:main:binary:0:111}, the contribution of the two leading edges cancel and we get $M_{\ref{gadget:main:binary:0:110}}^{-1} M_{\ref{gadget:main:binary:0:111}} = \diag(w,x,y,z)^{-1} \left(\shrinkMatrixTwoRows{\begin{bmatrix} w & x\\ y & z \end{bmatrix}^{\tensor 2}}\right)^{-1} \shrinkMatrixTwoRows{\begin{bmatrix} w & x\\ y & z \end{bmatrix}^{\tensor 2}} \diag(w,y,x,z) = \diag(1, y / x, x / y, 1)$. The resulting transition matrix has infinite order unless $x / y$ is a root of unity. This situation is analyzed formally in Lemma~\ref{lem:infOrder:xOverYNormNotOne}.

Another use of the anti-gadget technique can be applied with Gadgets~\ref{gadget:main:unary:0:111} and~\ref{gadget:main:unary:0:110}. Once again, the contribution of the leading edge cancels when composing an anti-gadget of Gadget~\ref{gadget:main:unary:0:110} with Gadget~\ref{gadget:main:unary:0:111}. The resulting matrix is a bit more complicated this time. However, when this pair of gadgets is analyzed formally in Lemma~\ref{lem:infOrder:x=0,wyz!=0}, the assumptions are $x = 0 \myand w y z \ne 0$. In that case, $M_{\ref{gadget:main:unary:0:110}}^{-1} M_{\ref{gadget:main:unary:0:111}} = \shrinkMatrixTwoRows{\begin{bmatrix} 1 & y^2 / w^2\\ 0 & 1 \end{bmatrix}}$. This matrix clearly has infinite order (up to a scalar).

\section{Interpolation Techniques}
The method of polynomial interpolation has been pioneered by Valiant~\cite{Val79} and further developed by many others~\cite{DG00, Vad01, BG05, BD07, CLX08}. In this section, we give a new unified technique to interpolate all unary signatures. This is our main technical step to prove $\SHARPP$-hardness. Our method produces an infinite set of pairwise linearly independent vectors at any fixed dimension, and then projects to a lower dimension while retaining pairwise linear independence of a nontrivial fraction.

In previous work, ``finisher gadgets''~\cite{KC10, CK10, CK11} were used to handle the \emph{symmetric} case, mapping \emph{symmetric} arity 2 signatures to arity 1 signatures. In the present work, we introduce \emph{projective gadget sets}. These gadget sets are completely general, in the sense that they can be used to map \emph{any} set of pairwise linearly independent signatures (symmetric \emph{or} asymmetric) to any lower arity, while preserving pairwise linear independence for an inverse polynomial fraction. This permits much more freedom in gadget constructions, and this power is used crucially in the proof of our dichotomy theorem. We remark that this advance is not just a simple matter of searching for the right gadgets. One must find the abstract criteria for success that simultaneously can be satisfied by gadgets that exist in practice. These developments, together with the anti-gadget concept, come together in the Group Lemma, which provides a straightforward criterion for proving $\SHARPP$-hardness of certain holant problems.

\begin{definition}
 A set of matrices $\mathcal{M}$ forms a \emph{projective set from arity $n$ to $m$} if for any matrix $N \in \C^{2^n \times 2}$ with rank 2, there exists a matrix $M \in \C^{2^m \times 2^n}$ in $\mathcal{M}$ such that $M N$ has rank 2.
\end{definition}

We also call a set of gadgets \emph{projective} from arity $n$ to $m$ if the set of its signature matrices is projective from arity $n$ to $m$. A gadget set that is projective from arity 2 to 1 can be used to transform a pair of $\GR$-gates with linearly independent binary signatures to a pair of $\GR$-gates with linearly independent unary signatures. Projector gadgets in such a set have 2 trailing edges and 1 leading edge, but can also be viewed as operating on signatures of higher arity, with the identity transformation being performed on the other edges not connected to the projector gadget. This way of connecting the projector gadget to an existing $\GR$-gate \emph{automatically} gives us projective gadget sets for higher arities. But first, a quick lemma to assist with the proof.

\begin{lemma} \label{lem:finish:helperForDownByOneWire}
 Let $v_0$ and $v_1$ be nonzero column vectors, not necessarily the same length. Then the block matrix $\shrinkMatrixTwoRows{\begin{bmatrix} av_0 & bv_0 \\ cv_1 & dv_1 \end{bmatrix}}$ has rank 2 if and only if $\shrinkMatrixTwoRows{\begin{bmatrix} a & b \\ c & d \end{bmatrix}}$ is invertible.
\end{lemma}
\begin{proof}
 We write $\shrinkMatrixTwoRows{\begin{bmatrix} av_0 & bv_0 \\ cv_1 & dv_1 \end{bmatrix}} = \shrinkMatrixTwoRows{\begin{bmatrix} v_0 & 0 \\ 0 & v_1 \end{bmatrix} \begin{bmatrix} a & b \\ c & d \end{bmatrix}}$, and $\shrinkMatrixTwoRows{\begin{bmatrix} v_0 & 0 \\ 0 & v_1 \end{bmatrix}}$ has rank 2. Thus $\shrinkMatrixTwoRows{\begin{bmatrix} av_0 & bv_0 \\ cv_1 & dv_1 \end{bmatrix}}$ has rank 2 if and only if $\shrinkMatrixTwoRows{\begin{bmatrix} a & b \\ c & d \end{bmatrix}}$ is invertible.
\end{proof}

\begin{lemma} \label{lem:finish:downByOneWire}
 Let $\mathcal{P}$ be a set of $\GR$-gadgets that is projective from arity $2$ to $1$. Then for all integers $k \ge 2$, $\mathcal{P}$ acts as a projective $\GR$-gadget set from arity $k$ to $k-1$.
\end{lemma}
\begin{proof}
 We are given that for any $N \in \C^{4 \times 2}$ with rank 2, there exists an $F \in \mathcal{P}$ such that $F \in \C^{2 \times 4}$ and $F N$ is invertible. We want to show that for any integer $k \ge 2$ and any rank 2 matrix $B \in \C^{2^k \times 2}$, there exists an $F \in \mathcal{P}$ such that $(I \tensor F) B$ has rank 2, where $I$ is the $2^{k - 2}$-by-$2^{k - 2}$ identity matrix.

 For any $F \in \mathcal{P}$, the matrix $I \tensor F$ can be viewed as being composed of 2-by-4 blocks, with $F$ appearing along the main diagonal and 2-by-4 zero-matrices elsewhere. We similarly view $B$ as being composed of 4-by-2 blocks $B = \shrinkMatrixFourRows{\begin{bmatrix} B_1 \\ B_2 \\ \vdots \\ B_{2^{k - 2}} \end{bmatrix}}$. Then $(I \tensor F) B = \shrinkMatrixFourRows{\begin{bmatrix} F B_1 \\ F B_2 \\ \vdots \\ F B_{2^{k - 2}} \end{bmatrix}} \in \C^{2^{k - 1} \times 2}$. If some $B_i$ has rank 2, then there is an $F \in \mathcal{P}$ such that $F B_i$ is invertible and $(I \tensor F) B$ has rank 2, as desired.

 Now assume otherwise, so each $B_i$ has rank at most 1. Since $B$ has rank 2, there exists a 2-by-2 invertible submatrix $D$ of $B$, for which the rows of $D$ appear in $B_i$ and $B_j$, for some $i < j$. It follows that $B_i$ and $B_j$ both have rank exactly 1. Hence for some nonzero vectors $v_0, v_1 \in \C^4$ and some $a, b, c, d \in \C$, we can write $B_i = [av_0\ bv_0]$ and $B_j = [cv_1\ dv_1]$. By Lemma~\ref{lem:finish:helperForDownByOneWire}, $\shrinkMatrixTwoRows{\begin{bmatrix} a & b \\ c & d \end{bmatrix}}$ is invertible, as $\shrinkMatrixTwoRows{\begin{bmatrix} B_i \\ B_j \end{bmatrix}}$ has rank 2. If $v_0$ and $v_1$ are linearly independent, then choose $F \in \mathcal{P}$ such that $F[v_0\ v_1]$ is invertible; otherwise let $\tilde{v} \in \C^4$ be such that $v_0$ and $\tilde{v}$ are linearly independent, and choose $F \in \mathcal{P}$ such that $F[v_0\ \tilde{v}]$ is invertible. In either case (ignoring $\tilde{v}$ in the second case), we define $[v_0'\ v_1'] = F[v_0\ v_1]$, where $v_0'$ and $v_1'$ are nonzero. Then by Lemma~\ref{lem:finish:helperForDownByOneWire}, the matrix $\shrinkMatrixTwoRows{\begin{bmatrix} av_0' & bv_0' \\ cv_1' & dv_1' \end{bmatrix}} = \shrinkMatrixTwoRows{\begin{bmatrix} F B_i \\ F B_j \end{bmatrix}}$ has rank 2, and since this appears as a submatrix of $(I \tensor F) B$, we are done.
\end{proof}

\begin{corollary} \label{cor:finish:downByAllWires}
 Let $\mathcal{P}$ be a finite set of $\GR$-gadgets that is projective from arity $2$ to $1$. Then for any integer $k \ge 2$, $\mathcal{P}$ induces a finite projective $\GR$-gadget set from arity $k$ to 1.
\end{corollary}

Now we show that a finite projective $\GR$-gadget set from arity $k$ to $1$ preserves pairwise linear independence for an inverse polynomial fraction of signatures. The essence of the next lemma is an exchange in the order of quantifiers.

\begin{lemma} \label{lem:finish:projectingEnough}
 Suppose $\{v_i\}_{i \ge 0}$ is a sequence of pairwise linearly independent column vectors in $\C^{2^k}$ and let $\mathcal{F} \subseteq \C^{2 \times 2^k}$ be a finite set of $f$ matrices that is projective from arity $k$ to 1. Then for every $n$, there exists some $F \in \mathcal{F}$ and some $S \subseteq \{F v_i \st 0 \le i \le n^f\}$ such that $|S| \ge n$ and the vectors in $S$ are pairwise linearly independent.
\end{lemma}
\begin{proof}
 Let $j > i \ge 0$ be integers and let $N = [v_i\ v_j] \in \C^{2^k \times 2}$. Since $v_i$ and $v_j$ are linearly independent, $\rank(N) = 2$. By assumption, there exists an $F \in \mathcal{F}$ such that $F N \in \C^{2 \times 2}$ is invertible, so we conclude that $F v_i$ and $F v_j$ are linearly independent.

 Each $F \in \mathcal{F}$ defines a coloring of the set $K = \{0, 1, \ldots, n^f\}$ as follows: color $i \in K$ with the linear subspace spanned by $F v_i$. Assume for a contradiction that for each $F \in \mathcal{F}$, there is not $n$ pairwise linearly independent vectors among $\{F v_i \st i \in K\}$. Then, including possibly the 0-dimensional subspace $\{0\}$, there can be at most $n$ distinct colors assigned by each $F \in \mathcal{F}$. By the pigeonhole principle, some $i$ and $j$ with $0 \le i < j \le n^f$ must receive the same color for all $F \in \mathcal{F}$. This is a contradiction with the previous paragraph, so we are done.
\end{proof}

The next lemma says that under suitable conditions, we can construct all unary signatures $(X,Y)$. The method will be interpolation at a higher dimensional iteration in a circular fashion and finishing with an appropriate \emph{projector} gadget.

\begin{lemma}[Group Lemma] \label{lem:interpolation}
 Let $\mathcal{P}$ be a finite set of projective $\GR$-gadgets from arity 2 to~1, and let $\mathcal{S}$ be a finite set of recursive $\GR$-gadgets of arity $d \ge 1$ with nonsingular transition matrices. Let $H$ be the group generated by the transition matrices of gadgets in $\mathcal{S}$, modulo scalar matrices $\lambda I$, for $\lambda \in \C - \{0\}$. If $H$ has infinite order, then any unary generator can be simulated: \[\text{For any $X, Y \in \C$, }\holant{\mathcal{G} \union \{(X, Y)\}}{\mathcal{R}} \le_T^{\P} \holant{\mathcal{G}}{\mathcal{R}}.\]
\end{lemma}
\begin{proof}
 We prove a weaker version by making the additional assumption that $\mathcal{G}$ contains a nondegenerate binary signature $g$, and $H$ contains some element of infinite order. The proof of the stronger version stated here is given in the appendix.

 Two matrices are unequal modulo scalar matrices $\lambda I$ if and only if they are linearly independent. If any member of $\mathcal{S}$, as a group element in $H$, has infinite order, then its powers supply an infinite set of pairwise linearly independent signatures. Otherwise they all have finite order, and the group $H$ is identical to the monoid generated by $\mathcal{S}$, i.e., every $h \in H$ is a product over $\mathcal{S}$ with non-negative powers. Such products give a composition of gadgets in $\mathcal{S}$, which is a recursive gadget. Let $h \in H$ have infinite order. Then the powers of $h$ supply an infinite set of pairwise linearly independent signatures.

 Before we can use a projector gadget set to project these signatures $\{h^i\}_{i \ge 0}$, we make a small modification to the gadget of $h^i$, for each $i$: connect a degree 2 vertex labeled with $g$ to every trailing edge. This ensures that the bipartite structure of the graph is preserved when applying projector gadgets. Let $M$ be the 2-by-2 matrix of $g$. As there are $d$ trailing edges, we apply $d$ copies of $g$, which corresponds to multiplication by the matrix $M^{\tensor d}$. Since $M$ is invertible, pairwise linear independence of the signatures is preserved. Now rewrite the $2^d$-by-$2^d$ matrix form of the signature $h^i M^{\tensor d}$ as a column vector $v_i \in \C^{2^{2 d}}$, indexed by $c_d \cdots c_1 b_1 \cdots b_d \in \{0,1\}^{2 d}$, where $b_1 \cdots b_d$ and $c_1 \cdots c_d$ are the row and column indices. Now we can attach projector gadgets to project each $v_i$ down to arity $1$ (see Figure~\ref{gadget:example:interpolationConstruction}).

 To show $\holant{\mathcal{G} \union \{(X, Y)\}}{\mathcal{R}} \le_T^{\P} \holant{\mathcal{G}}{\mathcal{R}}$, suppose we are given as input a bipartite signature grid $\Omega$ for $\holant{\mathcal{G} \union \{(X, Y)\}}{\mathcal{R}}$, with underlying graph $G = (V, E)$. Let $Q \subseteq V$ be the set of vertices labeled with generator $(X, Y)$, and let $n = |Q|$. By Corollary~\ref{cor:finish:downByAllWires}, there exists a finite projective set containing $f$ gadgets from arity $d$ to $1$, so by Lemma~\ref{lem:finish:projectingEnough} there is some projector gadget $F$ in this set such that at least $n + 2$ of the first $(n+2)^f + 1$ vectors of the form $F v_t$ are pairwise linearly independent. It is straightforward to efficiently find such a set; denote it by $S = \{(X_0, Y_0), (X_1, Y_1), \ldots, (X_{n+1}, Y_{n+1})\}$ and let $G_0, G_1, \ldots, G_{n+1}$ be the corresponding gadgets. At most one $Y_t$ can be zero, so without loss of generality assume $Y_t \ne 0$ for $0 \le t \le n$. If we replace every element of $Q$ with a copy of $G_t$, we obtain an instance of $\holant{\mathcal{G}}{\mathcal{R}}$ (note that the correct bipartite structure is preserved), and we denote this new signature grid by $\Omega_t$. Although $\Holant_{\Omega_t}$ is a sum of exponentially many terms, each nonzero term has the form $b X_t^i Y_t^{n-i}$ for some $i$ and for some $b \in \C$ that does not depend on $X_t$ or $Y_t$. Then for some $c_0, c_1, \ldots, c_n \in \C$, the sum can be rewritten as \[\Holant_{\Omega_t} = \sum_{0 \le i \le n} c_i X_t^i Y_t^{n-i}.\] Since each signature grid $\Omega_t$ is an instance of $\holant{\mathcal{G}}{\mathcal{R}}$, $\Holant_{\Omega_t}$ can be solved exactly using the oracle. Carrying out this process for every $t$ where $0 \le t \le n$, we arrive at a linear system where the $c_i$ values are the unknowns.
 \begin{eqnarray*}
  \begin{bmatrix}
   Y_{0}^{-n} \cdot \Holant_{\Omega_{0}} \\
   Y_{1}^{-n} \cdot \Holant_{\Omega_{1}} \\
   \vdots \\
   Y_{n}^{-n} \cdot \Holant_{\Omega_{n}}
  \end{bmatrix}
  &=&
  \begin{bmatrix}
   X_{0}^0 Y_{0}^0 & X_{0}^1 Y_{0}^{-1} & \cdots & X_{0}^n Y_{0}^{-n} \\
   X_{1}^0 Y_{1}^0 & X_{1}^1 Y_{1}^{-1} & \cdots & X_{1}^n Y_{1}^{-n} \\
   \vdots & \vdots & \ddots & \vdots \\
   X_{n}^0 Y_{n}^0 & X_{n}^1 Y_{n}^{-1} & \cdots & X_{n}^n Y_{n}^{-n} \\
  \end{bmatrix}
  \begin{bmatrix}
   c_0 \\
   c_1 \\
   \vdots \\
   c_n
  \end{bmatrix}
 \end{eqnarray*}
 The matrix above has entry $(X_r / Y_r)^c$ at row $r$ and column $c$. Due to pairwise linear independence of $(X_r, Y_r)$, $X_r / Y_r$ is pairwise distinct for $0 \le r \le n$. Hence this is a Vandermonde system of full rank, and we can solve it for the $c_i$ values. With these values in hand, we can calculate $\Holant_\Omega = \sum_{0 \le i \le n} c_i X^i Y^{n-i}$ directly, completing the reduction.
\end{proof}

Here is how we realize a projective set of gadgets from arity 2 to 1.
\begin{lemma} \label{lem:finish:sufficientCondition}
 Let $\Phi_i \in \C^{2 \times 4}$ for $1 \le i \le 7$ be matrices with the following properties: $\ker(\Phi_i) = \myspan{u, u_i}$ for $1 \le i \le 3$, $\ker(\Phi_{i + 3}) = \myspan{v, v_i}$ for $1 \le i \le 3$, $\ker(\Phi_7) = \myspan{s, t}$, and $\dim \{u, u_1, u_2, u_3\} = \dim \{v, v_1, v_2, v_3\} = \dim \{u, v, s, t\} = 4$. Then $\{\Phi_i \st 1 \le i \le 7\}$ is projective from arity 2 to 1.
\end{lemma}
\begin{proof}
 Let $N \in \C^{4 \times 2}$ be a rank two matrix with $\ker(N) = \myspan{w_1, w_2}$. If $\ker(N) = \myspan{u, v}$, then $\Phi_7 N$ has rank 2. Otherwise, either $\{w_1, w_2, u\}$ or $\{w_1, w_2, v\}$ is linearly independent. Say $\{w_1, w_2, u\}$ is linearly independent. Then $\{w_1, w_2, u\}$ can be further augmented by some $u_i$ for $1 \le i \le 3$ to form a basis, in which case $\Phi_i N$ has rank 2. The other case is similar.
\end{proof}

Verifying that a specific set of gadgets forms a projective set from arity 2 to 1 only requires a straightforward linear algebra computation. The proof of the following lemma is in the appendix. Note that the exceptional cases are either symmetric signatures (for which a dichotomy exists~\cite{KC10}) or largely correspond to tractable cases.
\begin{lemma} \label{lem:finish:generalPurposeBinary}
 There exists a finite projective $\GR$-gadget set from arity 2 to 1 unless \\\centerline{$x = y \myor w z = x y \myor (w, z) = (0, 0) \myor (x, y) = (0, 0) \myor (w^3 = -z^3 \myand x = -y)$}
\end{lemma}

Once we have all unary signatures at our disposal, we can prove $\SHARPP$-hardness under most settings. The proof for the following lemma is also in the appendix.
\begin{lemma} \label{lem:hardByUnary}
 Suppose $w,x,y,z \in \C$ and let $\mathcal{G}$ and $\mathcal{R}$ be finite signature sets with $(w,x,y,z) \in \mathcal{G}$ and $=_3 \in \mathcal{R}$. Assume that $\holant{\mathcal{G} \union \{(X_i, Y_i) \st 0 \le i < m\}}{\mathcal{R}} \le_T^{\P} \holant{\mathcal{G}}{\mathcal{R}}$ for any $X_i, Y_i \in \C$ and $m \in \Z^+$. Then $\holant{\mathcal{G}}{\mathcal{R}}$ is $\SHARPP$-hard unless $w z = x y \myor (w, z) = (0, 0) \myor (x, y) = (0, 0)$.
\end{lemma}

Combining the Group Lemma with Lemmas~\ref{lem:finish:generalPurposeBinary} and~\ref{lem:hardByUnary}, we get the following theorem.
\begin{theorem} \label{thm:interpolation}
 Suppose $w,x,y,z \in \C$ and let $\mathcal{G}$ and $\mathcal{R}$ be finite signature sets where $(w,x,y,z) \in \mathcal{G}$, $=_3 \in \mathcal{R}$, $x \ne y$, $w z \ne x y$, $(w, z) \ne (0, 0)$, $(x, y) \ne (0, 0)$, and $(w^3 \ne -z^3 \myor x \ne -y)$. Let $\mathcal{S}$ be a finite set of recursive $\GR$-gadgets of arity $d \ge 1$ with nonsingular transition matrices, and let $H$ be the group generated by the transition matrices of $\mathcal{S}$, modulo scalar matrices $\lambda I$, for $\lambda \in \C - \{0\}$. If $H$ has infinite order, then $\holant{\mathcal{G}}{\mathcal{R}}$ is $\SHARPP$-hard.
\end{theorem}

\section{Main Result}
\begin{theorem} \label{thm:dichotomy}
 Suppose $w, x, y, z \in \C$. Then $\holant{(w,x,y,z)}{{=}_3}$ is $\SHARPP$-hard except in the following classes, for which the problem is in $\P$.

 \emph{(1) degenerate:} $w z = x y$.

 \emph{(2) generalized disequality:} $w = z = 0$.

 \emph{(3) generalized equality:} $x = y = 0$.

 \emph{(4) affine after holographic transformation:} $w z = - x y \myand w^6 = \varepsilon z^6 \myand x^2 = \varepsilon y^2$, \emph{where} $\varepsilon = \pm 1$.

 \noindent
 If the input is restricted to planar graphs, then another class becomes tractable but everything else remains $\SHARPP$-hard.

 \emph{(5) computable by holographic algorithms with matchgates:} $w^3 = \varepsilon z^3 \myand x = \varepsilon y$, \emph{where} $\varepsilon = \pm 1$.
\end{theorem}

We prove the tractability part of Theorem~\ref{thm:dichotomy} next. The proof of $\SHARPP$-hardness begins in section~\ref{sec:anti-gadgets:main} and continues in section~\ref{sec:anti-gadgets:appendix} of the appendix.
\begin{proof}[Proof of tractability.]
 For any signature grid $\Omega$, $\Holant_{\Omega}$ is the product of the Holant on each connected component. Case~(1) is degenerate. We can break up every edge into two unary functions, and then the Holant value is a simple product over all vertices. For case~(2), on any connected component, the Holant value is zero unless it is bipartite, and if so a 2-coloring algorithm can be used to find the only two pertinent assignments, complements of each other. Similarly, for case~(3), only the all-0 and the all-1 assignments can possibly yield a nonzero value for each connected component. For case~(4), if $w = z = 0$, then this is already covered by case~(2). Otherwise $w z \ne 0$ since $w^6 = \varepsilon z^6$, in which case we apply the holographic transformation $\shrinkMatrixTwoRows{\begin{bmatrix} \alpha & 0\\ 0 & \alpha^2 \end{bmatrix}}$ with $\alpha = \varepsilon w^2 / z^2$. Note that $\alpha^3 = \varepsilon w^6 / z^6 = 1$. The edge signature becomes $(w,x,y,z) \shrinkMatrixTwoRows{\begin{bmatrix} \alpha & 0\\ 0 & \alpha^2 \end{bmatrix}^{\tensor 2}} = (\alpha^2 w, x, y, \alpha z)$, while $=_3$ is unchanged since $\shrinkMatrixTwoRows{\left(\begin{bmatrix} \alpha & 0\\ 0 & \alpha^2 \end{bmatrix}^{-1}\right)^{\tensor 3}} = I$, the 8-by-8 identity matrix~\cite{Val08, CL11}. This reduces to the case $w z = - x y \myand w^2 = \varepsilon z^2 \myand x^2 = \varepsilon y^2$, where $\varepsilon = \pm 1$. This edge signature belongs to the so-called affine function family and is tractable by Theorem 5.2 of~\cite{CLX09a}. Further discussion on this case is in section~\ref{sec:affine}. If the input is restricted to planar graphs and $w^3 = \varepsilon z^3 \myand x = \varepsilon y$, where $\varepsilon = \pm 1$, then we use the theory of holographic algorithms with matchgates to compute the Holant in polynomial time (see~\cite{CL11}).
\end{proof}

\section{Anti-Gadgets in Action} \label{sec:anti-gadgets:main}
Now we use our new idea of anti-gadgets to construct explicit matrices of infinite order.

\begin{lemma} \label{lem:infOrder:xOverYNormNotOne}
 If $w z \ne x y$, $w x y z \ne 0$, and $|x| \ne |y|$, then $\holant{(w,x,y,z)}{{=}_3}$ is $\SHARPP$-hard.
\end{lemma}
\begin{proof}
 The transition matrices for Gadgets~\ref{gadget:main:binary:0:110} and~\ref{gadget:main:binary:0:111} are $M_{\ref{gadget:main:binary:0:110}} = \shrinkMatrixTwoRows{\begin{bmatrix} w & x\\ y & z \end{bmatrix}^{\tensor 2}} \diag(w, x, y, z)$ and $M_{\ref{gadget:main:binary:0:111}} = \shrinkMatrixTwoRows{\begin{bmatrix} w & x\\ y & z \end{bmatrix}^{\tensor 2}} \diag(w, y, x, z)$, both nonsingular. Since the matrix $M_{\ref{gadget:main:binary:0:110}}^{-1} M_{\ref{gadget:main:binary:0:111}} = \diag(1, y/x, x/y, 1)$ has infinite order up to a scalar, we are done by Theorem~\ref{thm:interpolation}.
\end{proof}

\begin{lemma} \label{lem:infOrder:x=0,wyz!=0}
 If $x = 0$ and $w y z \ne 0$, then $\holant{(w,x,y,z)}{{=}_3}$ is $\SHARPP$-hard.
\end{lemma}
\begin{proof}
 The transition matrices for Gadget~\ref{gadget:main:unary:0:111} and Gadget~\ref{gadget:main:unary:0:110} are $M_{\ref{gadget:main:unary:0:111}} = \shrinkMatrixTwoRows{\begin{bmatrix} w & 0\\ y & z \end{bmatrix} \begin{bmatrix} w^2 & y^2\\ 0 & z^2 \end{bmatrix}}$ and $M_{\ref{gadget:main:unary:0:110}} = \shrinkMatrixTwoRows{\begin{bmatrix} w & 0\\ y & z \end{bmatrix} \begin{bmatrix} w^2 & 0\\ 0 & z^2 \end{bmatrix}}$, both nonsingular. Since $M_{\ref{gadget:main:unary:0:110}}^{-1} M_{\ref{gadget:main:unary:0:111}} = \shrinkMatrixTwoRows{\begin{bmatrix} 1 & y^2 / w^2\\ 0 & 1 \end{bmatrix}}$ has infinite order up to a scalar, we are done by Theorem~\ref{thm:interpolation}.
\end{proof}

\begin{lemma} \label{lem:infOrder:w=0,xyz!=0}
 If $w = 0$ and $x y z \ne 0$, then $\holant{(w,x,y,z)}{{=}_3}$ is $\SHARPP$-hard.
\end{lemma}
\begin{proof}
 The transition matrices for Gadgets~\ref{gadget:main:unary:0:110} and~\ref{gadget:main:unary:4:101010} are $M_{\ref{gadget:main:unary:0:110}} = \shrinkMatrixTwoRows{\begin{bmatrix} 0 & x\\ y & z \end{bmatrix} \begin{bmatrix} 0 & x y\\ x y & z^2 \end{bmatrix}}$ and $M_{\ref{gadget:main:unary:4:101010}} = \shrinkMatrixTwoRows{\begin{bmatrix} 0 & x\\ y & z \end{bmatrix}}$ $\shrinkMatrixTwoRows{\begin{bmatrix} 0 & x y z^3\\ x y z^3 & x y^2 z^2 + z^5 \end{bmatrix}}$, both nonsingular. Since $M_{\ref{gadget:main:unary:0:110}}^{-1} M_{\ref{gadget:main:unary:4:101010}} = z^3 \shrinkMatrixTwoRows{\begin{bmatrix} 1 & y / z\\ 0 & 1 \end{bmatrix}}$ has infinite order up to a scalar, we are done by Theorem~\ref{thm:interpolation}.
\end{proof}

The remainder of the proof of hardness for Theorem~\ref{thm:dichotomy} appears in section \ref{sec:anti-gadgets:appendix} of the appendix.

\section{Anti-Gadgets and Previous Work}
To further appreciate the usefulness of anti-gadgets, we show how this technique sheds new light on previous results.

One can find \emph{failure conditions} for a binary recursive gadget using the following lemma.
\begin{lemma} \label{lem:4RootsSameNormInText}
 Let $G$ be a binary recursive gadget having nonsingular transition matrix $M$. Then $\{M^i\}_{i \ge 0}$ is a sequence of pairwise linearly independent signatures unless $a_2 |a_1|^2 - |a_3|^2 \overline{a_2} a_0 = 0$, where $x^4 + a_3 x^3 + a_2 x^2 + a_1 x + a_0$ is the characteristic polynomial of $M$.
\end{lemma}
Analyzing a failure condition such as $a_2 |a_1|^2 - |a_3|^2 \overline{a_2} a_0 = 0$ simultaneously for several gadgets is quite difficult, even with the aid of symbolic computation. Previous work~\cite{CK10, CK11} relied heavily on miraculous cancellations in the failure conditions to contend with this. For example, consider the two gadgets in Figure \ref{fig:TAMC}. They are from~\cite{CK10}, where symmetric (i.e. $x=y$) signatures $(w,x,y,z)$ were considered on $k$-regular graphs.

After a change of variables $X = w z x^{-2}$ and $Y = (w / x)^3 + (z / x)^3$ and making a few assumptions to guarantee that $M_{\ref{gadget:TAMC:m1}}$ and $M_{\ref{gadget:TAMC:m2}}$ are nonsingular (which we omit in this discussion), the failure conditions of Gadgets~\ref{gadget:TAMC:m1} and~\ref{gadget:TAMC:m2} (when restricted to the real numbers) simplify to
\begin{align*}
 (X-1)^3 (X^{k-2} - 1) (X^{k-2} (X + 1)^2 (X^{k-1} + X^{k-2} + X + 3 Y + 1) - Y^3) &= 0, \\
 X^3 (X-1)^3 (X^{k-4} - 1) (X^{k-2} (X + 1)^2 (X^2 + X^{k-2} + X + 3 Y + X^{k-3}) - Y^3) &= 0.
\end{align*}
Assuming that both gadgets fail and $X \notin \{0, \pm 1\}$, this gives two polynomial expressions for $Y^3$. Setting these equal to each other and refactoring results in the contradiction $X^{k-2}(X + 1)^3 (X-1) (X^{k-3}-1) = 0$, implying that either one or the other gadget works. At the time of this discovery, it was a mystery whether there was any underlying explanation for such miraculous cancellations. Now we see how anti-gadgets reveal a better understanding of this same gadget pair.

By assuming that $M_{\ref{gadget:TAMC:m1}}$ \emph{fails} to produce an infinite set of pairwise linearly independent signatures, we have an explicit recursive gadget for $M_{\ref{gadget:TAMC:m1}}^{-1}$. Then $M_{\ref{gadget:TAMC:m1}}^{-1} M_{\ref{gadget:TAMC:m2}} = \diag(1,X,X,1)$ clearly produces an infinite set of pairwise linearly independent signatures unless $X$ is zero or a root of unity. Note that in the ``gadget language'' of $M_{\ref{gadget:TAMC:m1}}^{-1} M_{\ref{gadget:TAMC:m2}}$, the two leading directed edges of Gadget~\ref{gadget:TAMC:m1} and~\ref{gadget:TAMC:m2} simply annihilate each other, as do $k-4$ copies of the vertical edge. The signatures $=_3$ at the degree 3 vertices force the matrix $M_{\ref{gadget:TAMC:m1}}^{-1} M_{\ref{gadget:TAMC:m2}}$ to be diagonal. Thus, with almost no effort we have a strictly stronger result (i.e. over the complex numbers) through the use of an anti-gadget. This also shows that the anti-gadget concept is useful in the symmetric setting as well as the asymmetric setting.

In~\cite{CK11}, a similarly fantastic cancellation occurred involving Gadgets~\ref{gadget:COCOON:m1} and~\ref{gadget:COCOON:m2} (see Figure~\ref{fig:COCOON}). They form a suitable gadget and anti-gadget pair, as $M_{\ref{gadget:COCOON:m1}}^{-1} M_{\ref{gadget:COCOON:m2}}$ is a diagonal matrix. While this diagonal matrix is not as easy to analyze as the previous example, anti-gadgets would \emph{inform} the search for such useful gadgets, even if the analysis is carried out with different techniques.

\section*{Acknowledgements}
We thank Heng Guo for his many insightful comments and suggestions. We also thank him for pointing out an idea similar to that of an anti-gadget that appeared in the finite characteristic case of a very recent paper~\cite{GHLX11}, where finite order is forced by the characteristic. All authors were supported in part by NSF CCF-0914969.

\begin{figure}[ht]
 \centering
 \captionsetup[subfigure]{labelformat=empty}
 \gadgetSubfloat[Gadget \ref{gadget:TAMC:m1}]{
  \begin{tikzpicture}[scale=\scale,transform shape,>=\arrowType,node distance=\nodeDist,semithick]
   \node[external] (0)              {};
   \node[external] (1) [below of=0] {};
   \node[internal] (2) [right of=1] {};
   \node[internal] (3) [right of=0] {};
   \node[external] (4) [right of=2] {};
   \node[external] (5) [right of=3] {};
   \path (0) edge [->] node[pos=\startPos] (e1) {} (3)
         (1) edge [->] node[pos=\startPos] (e2) {} (2)
         (2) edge [<-, ultra thick]                (3)
             edge [-]  node[pos=\endPos]   (e3) {} (4)
         (3) edge [-]  node[pos=\endPos]   (e4) {} (5);
   \begin{pgfonlayer}{background}
    \node[draw=\borderColor,thick,rounded corners,fit = (2) (3) (e1) (e2) (e3) (e4)] {};
   \end{pgfonlayer}
  \end{tikzpicture}} \label{gadget:TAMC:m1}
 \qquad
 \gadgetSubfloat[Gadget \ref{gadget:TAMC:m2}]{
  \begin{tikzpicture}[scale=\scale,transform shape,>=\arrowType,node distance=\nodeDist,semithick]
   \node[external] (0)              {};
   \node[external] (1) [below of=0] {};
   \node[internal] (2) [right of=1] {};
   \node[internal] (3) [right of=0] {};
   \node[external] (4) [right of=2] {};
   \node[external] (5) [right of=3] {};
   \path (0) edge [->]             node[pos=\startPos] (e1) {} (3)
         (1) edge [->]             node[pos=\startPos] (e2) {} (2)
         (2) edge [<-, ultra thick]                            (3)
             edge [-]              node[pos=\endPos]   (e3) {} (4)
             edge [->, loop below] coordinate (c1)             (2)
         (3) edge [-]              node[pos=\endPos]   (e4) {} (5)
             edge [->, loop above] coordinate (c2)             (3);
   \begin{pgfonlayer}{background}
    \node[draw=\borderColor,thick,rounded corners,fit = (2) (3) (c1) (c2) (e1) (e2) (e3) (e4)] {};
   \end{pgfonlayer}
  \end{tikzpicture}} \label{gadget:TAMC:m2}
 \caption{Recursive gadgets from \cite{CK10} on $k$-regular graphs. Bold edges represent parallel edges. In Gadget \ref{gadget:TAMC:m1} (resp.~\ref{gadget:TAMC:m2}), the multiplicity is $k-2$ (resp.~$k-4$) so that the vertices have degree $k$.}
 \label{fig:TAMC}
\end{figure}
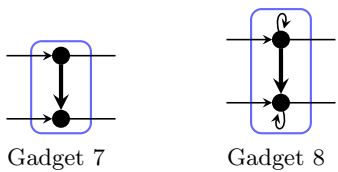

\begin{figure}[ht]
 \centering
 \captionsetup[subfigure]{labelformat=empty}
 \gadgetSubfloat[Gadget \ref{gadget:COCOON:m1}]{
  \begin{tikzpicture}[scale=\scale,transform shape,>=\arrowType,node distance=\nodeDist,semithick]
   \node[external] (0)              {};
   \node[internal] (1) [right of=0] {};
   \node[internal] (2) [above of=1] {};
   \node[external] (3) [right of=1] {};
   \path (0) edge [->]             node[pos=\startPos] (e1) {} (1)
         (1) edge [->, loop above]                             (1)
             edge [-]              node[pos=\endPos]   (e2) {} (3)
         (2) edge [<-, loop above] coordinate (c1)             (2)
             edge [<-, loop below]                             (2);
   \begin{pgfonlayer}{background}
    \node[draw=\borderColor,thick,rounded corners,fit = (c1) (e1) (e2)] {};
   \end{pgfonlayer}
  \end{tikzpicture}} \label{gadget:COCOON:m1}
 \qquad
 \gadgetSubfloat[Gadget \ref{gadget:COCOON:m2}]{
  \begin{tikzpicture}[scale=\scale,transform shape,>=\arrowType,node distance=\nodeDist,semithick]
   \node[external] (0)              {};
   \node[internal] (1) [right of=0] {};
   \node[internal] (2) [above of=1] {};
   \node[external] (3) [right of=1] {};
   \path (0) edge [->]             node[pos=\startPos] (e1) {} (1)
         (1) edge [->, bend left]                              (2)
             edge [<-, bend right]                             (2)
             edge [-]              node[pos=\endPos]   (e2) {} (3)
         (2) edge [<-, loop above] coordinate (c1)             (2);
   \begin{pgfonlayer}{background}
    \node[draw=\borderColor,thick,rounded corners,fit = (c1) (e1) (e2)] {};
   \end{pgfonlayer}
  \end{tikzpicture}} \label{gadget:COCOON:m2}
 \caption{Recursive gadgets from \cite{CK11} on $k$-regular graphs for $k$ even. The gadgets are pictured for $k = 4$ but generalize to all even $k \ge 4$ by adding self loops to the vertices.}
 \label{fig:COCOON}
\end{figure}
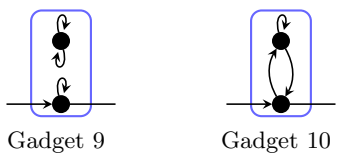

\bibliographystyle{plain}
\bibliography{bib}

\pagebreak
\appendix
\section*{Appendix}
\begin{figure}[h]
 \centering
 \def\decompScale{0.45}
 \def\dottedNodeDist{0.8cm}
 \begin{tikzpicture}[scale=\decompScale,transform shape,>=\arrowType,node distance=\nodeDist,semithick]
  \node[external]  (0)               {};
  \node[external]  (1) [right of= 0] {};
  \node[external]  (2) [right of= 1] {};
  \node[external]  (3) [right of= 2,node distance=\dottedNodeDist] {};
  \node[internal]  (4) [right of= 3] {};
  \node[external]  (5) [above right of=4,yshift=-13] {};
  \node[external]  (6) [below right of=4,yshift=13]  {};
  \node[external]  (7) [right of= 5,node distance=\dottedNodeDist] {};
  \node[external]  (8) [right of= 6,node distance=\dottedNodeDist] {};
  \node[external]  (9) [right of= 7] {};
  \node[external] (10) [right of= 8] {};
  \node[external] (11) [right of= 9] {};
  \node[external] (12) [right of=10] {};
  \node[external] (13) [right of=11,node distance=\dottedNodeDist] {};
  \node[external] (14) [right of=12,node distance=\dottedNodeDist] {};
  \node[internal] (15) [right of=13] {};
  \node[external] (16) [right of=14] {};
  \node[external] (18) [right of=15] {};
  \node[external] (17) [above of=18,yshift=-10] {};
   \path let
          \p1 = (15),
          \p2 = (17)
         in
          node[external,yshift=-5] (n1) at (\x1,\y2) {};
  \node[external] (19) [right of=16] {};
  \node[external] (20) [right of=17,node distance=\dottedNodeDist] {};
  \node[external] (21) [right of=18,node distance=\dottedNodeDist] {};
  \node[external] (22) [right of=19,node distance=\dottedNodeDist] {};
  \node[external] (23) [right of=20] {};
  \node[external] (24) [right of=21] {};
  \node[external] (25) [right of=22] {};
  \node[external] (26) [right of=23] {};
  \node[external] (27) [right of=24] {};
  \node[external] (28) [right of=25] {};
  \node[external] (29) [right of=26,node distance=\dottedNodeDist] {};
  \node[external] (30) [right of=27,node distance=\dottedNodeDist] {};
  \node[external] (31) [right of=28,node distance=\dottedNodeDist] {};
  \node[internal] (32) [right of=29] {};
  \node[external] (33) [right of=30] {};
  \node[external] (34) [right of=31] {};
  \node[external] (36) [right of=32] {};
  \node[external] (35) [above of=36,yshift=-10] {};
   \path let
          \p1 = (32),
          \p2 = (35)
         in
          node[external,yshift=-5] (n2) at (\x1,\y2) {};
  \node[external] (37) [right of=33] {};
  \node[external] (38) [right of=34] {};
  \node[external] (39) [right of=35,node distance=\dottedNodeDist] {};
  \node[external] (40) [right of=36,node distance=\dottedNodeDist] {};
  \node[external] (41) [right of=40,node distance=\dottedNodeDist] {};
   \path let
          \p1 = (41),
          \p2 = (39),
          \p3 = (40)
         in
          node[external] (42) at (\x1, \y2 / 2 + \y3 / 2) {};
   \path let
          \p1 = (42)
         in
          node[external] (n3) at (\x1, \y1 + 7) {};
   \path let
          \p1 = (42)
         in
          node[external] (n4) at (\x1, \y1 - 7) {};
   \path let
          \p1 = (42)
         in
          node[external] (n5) at (\x1 + 4, \y1) {};
  \node[external] (43) [right of=37] {};
  \node[external] (44) [right of=38] {};
  \node[external] (45) [right of=43] {};
  \node[external] (46) [right of=44] {};
  \node[external] (47) [right of=45,node distance=\dottedNodeDist] {};
  \node[external] (48) [right of=46,node distance=\dottedNodeDist] {};
  \node[external] (49) [right of=47] {};
  \node[external] (50) [right of=48] {};
  \node[external] (51) [right of=49] {};
  \node[external] (52) [right of=50] {};
  \node[external] (53) [right of=51,node distance=\dottedNodeDist] {};
  \node[external] (54) [right of=52,node distance=\dottedNodeDist] {};
  \node[internal] (55) [below right of=53,yshift=13] {};
  \node[external] (56) [right of=55] {};
  \path (0) edge [->]                    node[pos=0.3] (e1) {} node[pos=0.7] (e2) {}    (2)
        (2) edge [-, dotted, very thick]                                                (3)
        (3) edge [-]                     node[pos=\startPos] (e3) {}                    (4)
        (4) edge [-, out=90, in=180]     node[pos=0.4]       (e4) {}                    (5)
            edge [-, out=-90, in=180]    node[pos=0.4]       (e5) {}                    (6)
        (5) edge [-, dotted, very thick]                                                (7)
        (6) edge [-, dotted, very thick]                                                (8)
        (7) edge [->]                    node[pos=0.3] (e6) {} node[pos=0.7] (e7) {}   (11)
        (8) edge [-]                     node[pos=0.3] (e8) {} node[pos=0.7] (e9) {}   (12)
       (11) edge [-, dotted, very thick]                                               (13)
       (12) edge [-, dotted, very thick]                                               (14)
       (13) edge [-]                     node[pos=\startPos] (e10) {}                  (15)
       (14) edge [-]                     node[pos=0.3] (e11) {} node[pos=0.7] (e12) {} (19)
       (15) edge [-, bend left]          node[pos=\endPos]   (e13) {}                  (17)
            edge [-]                     node[pos=\endPos]   (e14) {}                  (18)
       (17) edge [-, dotted, very thick]                                               (20)
       (18) edge [-, dotted, very thick]                                               (21)
       (19) edge [-, dotted, very thick]                                               (22)
       (20) edge [<-]                    node[pos=0.3] (e15) {} node[pos=0.7] (e16) {} (26)
       (21) edge [-]                     node[pos=0.3] (e17) {} node[pos=0.7] (e18) {} (45)
       (22) edge [-]                     node[pos=0.3] (e19) {} node[pos=0.7] (e20) {} (46)
       (26) edge [-, dotted, very thick]                                               (29)
       (29) edge [-]                     node[pos=\startPos] (e21) {}                  (32)
       (32) edge [-, bend left]          node[pos=\endPos]   (e22) {}                  (35)
            edge [-]                     node[pos=\endPos]   (e23) {}                  (36)
       (35) edge [-, dotted, very thick]                                               (39)
       (36) edge [-, dotted, very thick]                                               (40)
       (39) edge [-, out=0, in=90]                                                     (42.center)
(42.center) edge [->, out=-90, in=0]                                                   (40)
       (45) edge [-, dotted, very thick]                                               (47)
       (46) edge [-, dotted, very thick]                                               (48)
       (47) edge [->]                    node[pos=0.3] (e24) {} node[pos=0.7] (e25) {} (51)
       (48) edge [<-]                    node[pos=0.3] (e26) {} node[pos=0.7] (e27) {} (52)
       (51) edge [-, dotted, very thick]                                               (53)
       (52) edge [-, dotted, very thick]                                               (54)
       (55) edge [-, out=90, in=0]       node[pos=0.4]     (e29) {}                    (53)
            edge [-, out=-90, in=0]      node[pos=0.4]     (e30) {}                    (54)
            edge [-]                     node[pos=\endPos] (e31) {}                    (56);
  \begin{pgfonlayer}{background}
   \node[draw=\borderColor,thick,rounded corners,fit = (e1) (e2)] {};
   \node[draw=\borderColor,thick,rounded corners,fit = (e3) (e4) (e5)] {};
   \node[draw=\borderColor,thick,rounded corners,fit = (e6) (e7) (e8) (e9)] {};
   \node[draw=\borderColor,thick,rounded corners,fit = (n1) (e10) (e11) (e12) (e13) (e14)] {};
   \node[draw=\borderColor,thick,densely dashed,rounded corners,fit = (e15) (e16)] {};
   \node[draw=\borderColor,thick,densely dashed,rounded corners,fit = (n2) (e21) (e22) (e23)] {};
   \node[draw=\borderColor,thick,densely dashed,rounded corners,fit = (42) (n3) (n4)] {};
   \node[draw=\borderColor,thick,rounded corners,fit = (35) (39) (n5) (e15) (e17) (e18) (e19) (e20)] {};
   \node[draw=\borderColor,thick,rounded corners,fit = (e24) (e25) (e26) (e27)] {};
   \node[draw=\borderColor,thick,rounded corners,fit = (e29) (e30) (e31)] {};
  \end{pgfonlayer}
 \end{tikzpicture}
 \caption{Composition of Gadget \ref{gadget:main:unary:4:101010} using the basic gadget components in Figure \ref{gadget:simple}}
 \label{fig:gadget:decomposition:unary:4:101010}
\end{figure}
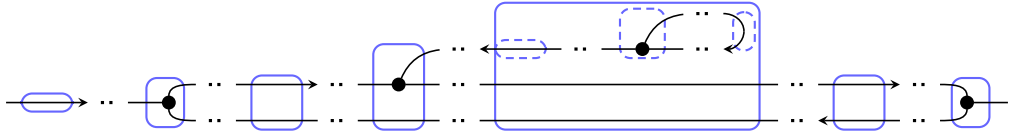

\section{Proof of Group Lemma}
\begin{proof}
 Two matrices are unequal modulo scalar matrices $\lambda I$ if and only if they are linearly independent. If any member of $\mathcal{S}$, as a group element in $H$, has infinite order, then its powers supply an infinite set of pairwise linearly independent signatures. Otherwise they all have finite order, and the group $H$ is identical to the monoid generated by $\mathcal{S}$, i.e., every $h \in H$ is a product over $\mathcal{S}$ with non-negative powers. Such products give a composition of gadgets in $\mathcal{S}$, which is a recursive gadget. By assumption, $H$ has infinite order, so by composing recursive gadgets from $\mathcal{S}$, a breadth-first traversal of the Cayley graph of the monoid generated by $\mathcal{S}$ supplies an arbitrarily large set of recursive gadgets having pairwise linearly independent signatures.

 Before we can use a projective gadget set to project the set of pairwise linearly independent signatures down to arity $1$, we make a small modification to each corresponding gadget: connect a nondegenerate generator $g$ to every trailing edge. This ensures that the bipartite structure of the graph is preserved when applying projector gadgets. We claim that there is some nondegenerate signature $g \in \mathcal{G}$. If this were not the case, then any recursive gadget $s \in \mathcal{S}$ (note $\mathcal{S}$ is nonempty) could be rewritten with all leading edges internally incident to unary signatures. The recurrence matrix of such a gadget is expressible as a product of a column vector and a row vector (by partitioning $s$ into two gadgets with no shared edges), hence the recurrence matrix of $s$ would have rank at most 1, which is less than $2^d$ as promised. Let $a \ge 2$ be the arity of $g$. One can show by induction that any nondegenerate signature has at least one index $i$, such that if we express the signature as a $2$-by-$2^{a-1}$ matrix $M$ indexed by the $i$-th variable for the row and the remaining $a-1$ variables for the column, then $M$ has rank 2. We designate one such dangling edge of $g$ as the leading edge and all other dangling edges as trailing edges. As there are $d$ trailing edges in $s$, we apply $d$ copies of $g$, which corresponds to multiplication by the matrix $M^{\tensor d}$. Since $M$ has full rank, pairwise linear independence of the signatures is preserved. Now rewrite the $2^d$-by-$2^{d (a-1)}$ matrix form of the signature as a column vector in $\C^{2^{d a}}$, indexed by $c_{d (a-1)} \cdots c_1 b_1 \cdots b_d \in \{0,1\}^{d a}$, where $b_1 \cdots b_d$ and $c_1 \cdots c_{d (a-1)}$ are the row and column indices. Denote these vectors as $\{v_i\}_{i \ge 0}$. Finally we can attach projector gadgets to project each $v_i$ down to arity $1$.

 To show $\holant{\mathcal{G} \union \{(X, Y)\}}{\mathcal{R}} \le_T^{\P} \holant{\mathcal{G}}{\mathcal{R}}$, suppose we are given as input a bipartite signature grid $\Omega$ for $\holant{\mathcal{G} \union \{(X, Y)\}}{\mathcal{R}}$, with underlying graph $G = (V, E)$. Let $Q \subseteq V$ be the set of vertices labeled with generator $(X, Y)$, and let $n = |Q|$. By Corollary~\ref{cor:finish:downByAllWires}, there exists a finite projective set containing $f$ gadgets from arity $d$ to $1$, so by Lemma~\ref{lem:finish:projectingEnough} there is some projector gadget $F$ in this set such that at least $n+2$ of the first $(n+2)^f + 1$ vectors of the form $F v_t$ are pairwise linearly independent. It is straightforward to efficiently find such a set; denote it by $S = \{(X_0, Y_0), (X_1, Y_1), \ldots, (X_{n+1}, Y_{n+1})\}$ and let $G_0, G_1, \ldots, G_{n+1}$ be the corresponding gadgets. At most one $Y_t$ can be zero, so without loss of generality assume $Y_t \ne 0$ for $0 \le t \le n$. If we replace every element of $Q$ with a copy of $G_t$, we obtain an instance of $\holant{\mathcal{G}}{\mathcal{R}}$ (note that the correct bipartite structure is preserved), and we denote this new signature grid by $\Omega_t$. Although $\Holant_{\Omega_t}$ is a sum of exponentially many terms, each nonzero term has the form $b X_t^i Y_t^{n-i}$ for some $i$ and for some $b \in \C$ that does not depend on $X_t$ or $Y_t$. Then for some $c_0, c_1, \dots, c_n \in \C$, the sum can be rewritten as \[\Holant_{\Omega_t} = \sum_{0 \le i \le n} c_i X_t^i Y_t^{n-i}.\] Since each signature grid $\Omega_t$ is an instance of $\holant{\mathcal{G}}{\mathcal{R}}$, $\Holant_{\Omega_t}$ can be solved exactly using the oracle. Carrying out this process for every $t$ where $0 \le t \le n$, we arrive at a linear system where the $c_i$ values are the unknowns.
 \begin{eqnarray*}
  \begin{bmatrix}
   Y_{0}^{-n} \cdot \Holant_{\Omega_{0}} \\
   Y_{1}^{-n} \cdot \Holant_{\Omega_{1}} \\
   \vdots \\
   Y_{n}^{-n} \cdot \Holant_{\Omega_{n}}
  \end{bmatrix}
  &=&
  \begin{bmatrix}
   X_{0}^0 Y_{0}^0 & X_{0}^1 Y_{0}^{-1} & \cdots & X_{0}^n Y_{0}^{-n} \\
   X_{1}^0 Y_{1}^0 & X_{1}^1 Y_{1}^{-1} & \cdots & X_{1}^n Y_{1}^{-n} \\
   \vdots & \vdots & \ddots & \vdots \\
   X_{n}^0 Y_{n}^0 & X_{n}^1 Y_{n}^{-1} & \cdots & X_{n}^n Y_{n}^{-n} \\
  \end{bmatrix}
  \begin{bmatrix}
   c_0 \\
   c_1 \\
   \vdots \\
   c_n
  \end{bmatrix}
 \end{eqnarray*}
 The matrix above has entry $(X_r / Y_r)^c$ at row $r$ and column $c$. Due to pairwise linear independence of $(X_r, Y_r)$, $X_r / Y_r$ is pairwise distinct for $0 \le r \le n$. Hence this is a Vandermonde system of full rank, and we can solve it for the $c_i$ values. With these values in hand, we can calculate $\Holant_\Omega = \sum_{0 \le i \le n} c_i X^i Y^{n-i}$ directly, completing the reduction.
\end{proof}

\section{Proof of Lemma \ref{lem:finish:generalPurposeBinary}}
We are given that $x \ne y \myand w z \ne x y \myand (w, z) \ne (0, 0) \myand (x, y) \ne (0, 0) \myand (w^3 \ne -z^3 \myor x \ne -y)$. We prove Lemma~\ref{lem:finish:generalPurposeBinary} by exhibiting projective gadget sets that satisfy the hypotheses of Lemma~\ref{lem:finish:sufficientCondition}. Let $F_i$ be the transition matrix of Gadget $i$ for $\ref{gadget:finish:0:1} \le i \le \ref{gadget:finish:27:1010011}$. There are five cases of projective $\GR$-gadget sets from arity 2 to 1. We omit the verification that each set of projectors forms a projective gadget set from arity 2 to 1 under its particular assumptions since this is a straightforward linear algebra computation. The five cases are

(1) $w z \ne x y \myand w x y z \ne 0 \myand w^3 x + w x y z + w^2 z^2 + y z^3 \ne 0 \myand x^2 \ne y^2$,

(2) $w z \ne x y \myand w x y z \ne 0 \myand w^3 x + w x y z + w^2 z^2 + y z^3 \ne 0 \myand x = -y \myand w^3 \ne -z^3$,

(3) $w z \ne x y \myand w x y z \ne 0 \myand w^3 x + w x y z + w^2 z^2 + y z^3 = 0 \myand x \ne y$,

(4) $w z \ne x y \myand w = 0 \myand z \ne 0 \myand x \ne y$, and

(5) $w z \ne x y \myand x = 0 \myand y \ne 0$.

\noindent
Which projectors are used in each case (and the role of each projector within each case) can be found in Table~\ref{tbl:finishers:generalPurposeBinary}. In all five cases, the vector $u$ in the kernels of $\Phi_1$, $\Phi_2$, and $\Phi_3$ is $(0, -1, 1, 0)$ and the vector $v$ in the kernels of $\Phi_4$, $\Phi_5$, and $\Phi_6$ is $(0, -x, y, 0)$.

All five cases utilize the assumption $w z \ne x y$, i.e., the edge signature is non-degenerate. Under three additional disequality assumptions, the projectors in row~1 of Table~\ref{tbl:finishers:generalPurposeBinary} have the desired properties. The purpose of the remaining four cases is to handle the situation that these three disequalities are not all true.

Case~(2) retains two of the additional disequality assumptions but assumes that $x^2 = y^2$. Since we are considering the asymmetric case, the only option is $x = -y$. By assumption, it is not the case that $x = -y \myand w^3 = -z^3$, so we have $w^3 \ne -z^3$. Under these conditions, the projectors in row~2 of Table~\ref{tbl:finishers:generalPurposeBinary} have the desired properties.

Like cases~(1) and~(2), case~(3) retains the assumption that no variable is zero but now considers the case that the polynomial $w^3 x + w x y z + w^2 z^2 + y z^3$ is zero. Given that we are also considering the asymmetric case, i.e., $x \not = y$, the projectors in row~3 of Table~\ref{tbl:finishers:generalPurposeBinary} have the desired properties.

Cases~(4) and~(5) handle the remaining case $w z \ne x y \myand w x y z = 0$. The assumptions $w z \ne x y \myand (w,z) \ne (0,0) \myand (x,y) \ne (0,0)$ imply that at most one of $w$, $x$, $y$, and $z$ is zero. By switching the role of 0 and 1 via the holographic transformation $\shrinkMatrixTwoRows{\begin{bmatrix} 0 & 1\\ 1 & 0\end{bmatrix}}$, the complexity of the case $y z = 0$ is the same as the complexity of the case $w x = 0$. Therefore, we assume that $y z \ne 0$. Case~(4) considers $w$ as zero, so $z$ is nonzero by assumption. Then still within the asymmetric case, the projectors in row~4 of Table~\ref{tbl:finishers:generalPurposeBinary} have the desired properties. Case~(5) considers $x$ as zero, so $y$ is nonzero by assumption and the projectors in row~5 of Table~\ref{tbl:finishers:generalPurposeBinary} have the desired properties.

These five cases cover all settings not excluded by the assumptions in the statement of the lemma, so the proof is complete.
\begin{table}[t]
 \centering
 \begin{align*}
  \begin{array}{|c||c|c|c||c|c|c||c|}\hline
                 & \Phi_1                                         & \Phi_2                                            & \Phi_3                                                & \Phi_4                                         & \Phi_5                                            & \Phi_6                             & \Phi_7                                            \\\hline\hline
   \text{Case 1} & \multirow{5}{*}{$F_{\ref{gadget:finish:0:1}}$} & F_{\ref{gadget:finish:13:1000000}}                & \multirow{2}{*}{$F_{\ref{gadget:finish:13:1010100}}$} & \multirow{5}{*}{$F_{\ref{gadget:finish:0:1}}$} & \multirow{5}{*}{$F_{\ref{gadget:finish:2:1000}}$} & F_{\ref{gadget:finish:2:1011}}     & \multirow{2}{*}{$F_{\ref{gadget:finish:3:1010}}$} \\\cline{1-1}\cline{3-3}\cline{7-7}
   \text{Case 2} &                                                & \multirow{3}{*}{$F_{\ref{gadget:finish:2:1010}}$} &                                                       &                                                &                                                   & F_{\ref{gadget:finish:13:1000100}} &                                                   \\\cline{1-1}\cline{4-4}\cline{7-8}
   \text{Case 3} &                                                &                                                   & F_{\ref{gadget:finish:27:1010011}}                    &                                                &                                                   & F_{\ref{gadget:finish:27:1000010}} & F_{\ref{gadget:finish:14:1001010}}                \\\cline{1-1}\cline{4-4}\cline{7-8}
   \text{Case 4} &                                                &                                                   & \multirow{2}{*}{$F_{\ref{gadget:finish:19:1010100}}$} &                                                &                                                   & F_{\ref{gadget:finish:19:1000110}} & \multirow{2}{*}{$F_{\ref{gadget:finish:3:1000}}$} \\\cline{1-1}\cline{3-3}\cline{7-7}
   \text{Case 5} &                                                & F_{\ref{gadget:finish:19:1010111}}                &                                                       &                                                &                                                   & F_{\ref{gadget:finish:2:1011}}     &                                                   \\\hline
  \end{array}
 \end{align*}
 \caption{This table indicates which projectors are used in each case (and the role of each projector within each case) in the proof of Lemma~\ref{lem:finish:generalPurposeBinary}. The seven $\Phi_i$ refer to the matrices in Lemma~\ref{lem:finish:sufficientCondition}. As an example, the projective set in case~(1) is $\{F_{\ref{gadget:finish:0:1}}, F_{\ref{gadget:finish:13:1000000}}, F_{\ref{gadget:finish:13:1010100}}, F_{\ref{gadget:finish:0:1}}, F_{\ref{gadget:finish:2:1000}}, F_{\ref{gadget:finish:2:1011}}, F_{\ref{gadget:finish:3:1010}}\}$. Note that $F_{\ref{gadget:finish:0:1}}$ plays the role of both $\Phi_1$ and $\Phi_4$.}
 \label{tbl:finishers:generalPurposeBinary}
\end{table}

\section{Proof of Lemma \ref{lem:hardByUnary}} \label{sec:hardByUnaryProof}
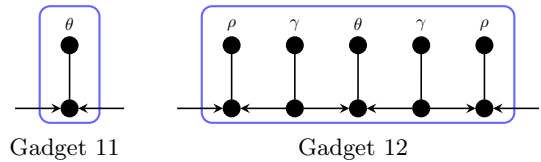
\begin{figure}[b]
 \centering
 \captionsetup[subfigure]{labelformat=empty}
 \gadgetSubfloat[Gadget \ref{gadget:oneWeed}]{
  \begin{tikzpicture}[scale=\scale,transform shape,>=\arrowType,node distance=\nodeDist,semithick]
   \node[external] (0)              {};
   \node[internal] (1) [right of=0] {};
   \node[internal] (2) [above of=1] {};
   \node[external] (3) [right of=1] {};
   \node[above] at (2.north) (l1) {$\theta$};
   \path (0) edge [->] node[pos=\startPos] (e1) {} (1)
         (1) edge [-]                              (2)
             edge [<-] node[pos=\endPos]   (e2) {} (3);
   \begin{pgfonlayer}{background}
    \node[draw=\borderColor,thick,rounded corners,fit = (1) (2) (e1) (e2) (l1)] {};
   \end{pgfonlayer}
  \end{tikzpicture}} \label{gadget:oneWeed}
 \gadgetSubfloat[Gadget \ref{gadget:threeWeeds}]{
  \begin{tikzpicture}[scale=\scale,transform shape,>=\arrowType,node distance=\nodeDist,semithick]
   \node[external] (0)               {};
   \node[internal] (1)  [right of=0] {};
   \node[internal] (2)  [above of=1] {};
   \node[internal] (3)  [right of=1] {};
   \node[internal] (4)  [above of=3] {};
   \node[internal] (5)  [right of=3] {};
   \node[internal] (6)  [above of=5] {};
   \node[internal] (7)  [right of=5] {};
   \node[internal] (8)  [above of=7] {};
   \node[internal] (9)  [right of=7] {};
   \node[internal] (10) [above of=9] {};
   \node[external] (11) [right of=9] {};
   \node[above] at (2.north)  (l1) {$\rho$};
   \node[above] at (4.north)  (l2) {$\gamma$};
   \node[above] at (6.north)  (l3) {$\theta$};
   \node[above] at (8.north)  (l4) {$\gamma$};
   \node[above] at (10.north) (l5) {$\rho$};
   \path (0) edge [->] node[pos=\startPos] (e1) {}  (1)
         (1) edge [-]                               (2)
             edge [<-]                              (3)
         (3) edge [-]                               (4)
             edge [->]                              (5)
         (5) edge [-]                               (6)
             edge [<-]                              (7)
         (7) edge [-]                               (8)
             edge [->]                              (9)
         (9) edge [-]                              (10)
             edge [<-] node[pos=\endPos]   (e2) {} (11);
   \begin{pgfonlayer}{background}
    \node[draw=\borderColor,thick,rounded corners,fit = (1) (2) (3) (4) (5) (6) (7) (8) (9) (10) (e1) (e2) (l1) (l2) (l3) (l4) (l5)] {};
   \end{pgfonlayer}
  \end{tikzpicture}} \label{gadget:threeWeeds}
 \caption{Gadgets used to simulate the generator $(0,1,1,1)$}
\end{figure}

The proof of Lemma~\ref{lem:hardByUnary} makes use of the following lemma.
\begin{lemma}[Lemma 3.3 of \cite{KC10}] \label{lem:interpolation:symmetricSimulationOfVertexCover}
 Suppose that $(a,b) \in \C^2 - \{(a,b) \st a b = 1\} - (0,0)$ and let $\mathcal{G}$ and $\mathcal{R}$ be finite signature sets where $(a,1,1,b) \in \mathcal{G}$ and $=_3 \in \mathcal{R}$. Further assume that $\holant{\mathcal{G} \union \{(X_i, Y_i)\} \st 0 \le i < m\}}{\mathcal{R}} \le_T^{\P} \holant{\mathcal{G}}{\mathcal{R}}$ for any $X_i, Y_i \in \C$ and $m \in \Z^+$. Then $\holant{\mathcal{G} \union \{(0,1,1,1)\}}{\mathcal{R}} \le_T^{\P} \holant{\mathcal{G}}{\mathcal{R}}$ and $\holant{\mathcal{G}}{\mathcal{R}}$ is $\SHARPP$-hard.
\end{lemma}

\begin{proof}[Proof of Lemma \ref{lem:hardByUnary}.]
 Since $\holant{(0,1,1,1)}{{=}_3}$, \#\textsc{VertexCover} on 3-regular graphs, is $\SHARPP$-hard, we only need to show how to simulate the generator signature $(0,1,1,1)$. The assumptions $w z \ne x y \myand (w,z) \ne (0,0) \myand (x,y) \ne (0,0)$ imply that at most one of $w$, $x$, $y$, and $z$ is zero. By switching the role of 0 and 1 via the holographic transformation $\shrinkMatrixTwoRows{\begin{bmatrix} 0 & 1\\ 1 & 0\end{bmatrix}}$, the complexity of the case $y z = 0$ is the same as the complexity of the case $w x = 0$. Therefore, we assume that $y z \ne 0$.

 If $w = 0$, then Gadget~\ref{gadget:oneWeed} with $\theta = \frac{1}{x} \left(\frac{z}{y^2}, \frac{1}{z}\right)$ simulates $(x / z, 1, 1, 2 z / x)$, which can in turn simulate $(0,1,1,1)$ by Lemma~\ref{lem:interpolation:symmetricSimulationOfVertexCover}. If $x = 0$, then Gadget~\ref{gadget:oneWeed} with $\theta = \frac{1}{w} \left(\frac{1}{y}, \frac{y}{z^2}\right)$ simulates $(w / y, 1, 1, 2 y / w)$, which can in turn simulate $(0,1,1,1)$ by Lemma~\ref{lem:interpolation:symmetricSimulationOfVertexCover}. If $w x \neq 0 \myand w z = -x y$, then Gadget~\ref{gadget:oneWeed} with $\theta = \frac{1}{x y} \left(\frac{2 x}{w}, \frac{w}{x}\right)$ simulates $(3 w / y, 1, 1, 3 y / w)$, which can in turn simulate $(0,1,1,1)$ by Lemma~\ref{lem:interpolation:symmetricSimulationOfVertexCover}. Finally if $w x \ne 0 \myand w z \ne x y \myand w z \ne -x y$, then Gadget~\ref{gadget:threeWeeds} with $\theta = \frac{w z + x y}{w x(w z - x y)} \left(\frac{-x}{w}, \frac{w}{x}\right)$, $\gamma = \frac{1}{(w z - x y)}\left(\frac{-1}{w x}, \frac{w x}{y z (w z + x y)}\right)$, and $\rho = (x z, -w y)$ simulates $(0,1,1,1)$.
\end{proof}

\section{More Anti-Gadgets in Action} \label{sec:anti-gadgets:appendix}
For the remainder of the proof of $\SHARPP$-hardness of Theorem~\ref{thm:dichotomy}, we use our anti-gadget technique in combination with Lemmas~\ref{lem:2RootsSameNorm}, \ref{lem:4RootsSameNorm}, and~\ref{lem:8RootsSameNorm}. In the contrapositive, these lemmas provide sufficient conditions to conclude that a matrix has infinite order (up to a scalar). Their proofs follow from a few observations. For monic polynomials in $\C[X]$ of degree $n$ with roots $\lambda_i$ for $1 \le i \le n$ of the same nonnegative norm $r \in \R$, let $a_k \in \C$ be the coefficient of $X^k$ and $\sigma_k$ the elementary symmetric polynomial of degree $k$ in $\lambda_i / r$ for $1 \le i \le n$, the norm one (scaled) roots.\footnote{This argument assumes $r \not = 0$.
However, when  $r = 0$, the conclusion still holds trivially.} Thus, $a_k = (-r)^{n-k} \sigma_{n-k}$. By being norm one, $\sigma_k = \overline{\sigma_{n-k}} \sigma_n$, $a_k = (-1)^n r^{n - 2 k} \overline{a_{n-k}} \sigma_n$, and $|a_k| = r^{n - 2 k} |a_{n-k}|$,
for $0 \le k < n$.
\begin{lemma}[Lemma 4.4 in \cite{KC10}] \label{lem:2RootsSameNorm}
 If both roots of $X^2 + a_1 X + a_0 \in \C[X]$ have the same norm, then $a_1 |a_0| = \overline{a_1} a_0$. If further $a_0 a_1 \ne 0$, then $\operatorname{Arg}(a_1^2) = \operatorname{Arg}(a_0)$ thus $a_1^2 / a_0 \in \R^+$.
\end{lemma}

\begin{lemma} \label{lem:4RootsSameNorm}
 If all roots of $X^4 + a_3 X^3 + a_2 X^2 + a_1 X + a_0 \in \C[X]$ have the same norm, then $a_2 |a_1|^2 = |a_3|^2 \overline{a_2} a_0$.
\end{lemma}

\begin{lemma} \label{lem:8RootsSameNorm}
 If $\sum_{k=0}^8 a_k X^k \in \C[X]$ is monic and all roots have the same norm, then $a_3^2 |a_1|^2 = |a_7|^2 \overline{a_5}^2 a_0^2$, $a_4 |a_2|^2 = |a_6|^2 \overline{a_4} a_0$, and $|a_3|^2 a_2 = \overline{a_6} |a_5|^2 a_0$.
\end{lemma}

On directed 3-regular graphs, there are some symmetries under which the Holant is invariant. The next lemma states these symmetries.
\begin{lemma} \label{lem:algebraicSymmetrization}
Let $G$ be a directed 3-regular graph. Then there exists a polynomial $P$ with integer coefficients in six variables, such that for any signature grid $\Omega$ having underlying graph $G$ with vertex signature $=_3$ and edge signature $(w,x,y,z)$, the Holant value is
 \begin{align*}
  \operatorname{Holant}_\Omega = P (w z, x y, w^3 + z^3, x + y, w^3 x + y z^3, w^3 y + x z^3).
 \end{align*}
\end{lemma}

\begin{proof}
 Consider any ${0, 1}$ vertex assignment $\sigma$ with a non-zero valuation. If $\sigma'$ is the complement assignment switching all 0's and 1's in $\sigma$, then for $\sigma$ and $\sigma'$, we have the sum of valuations $w^a x^b y^c z^d + w^d x^c y^b z^a$ for some $a,b,c,d$. Here $a$ (resp.~$d$) is the number of edges connecting two degree 3 vertices both assigned~0 (resp.~1) by $\sigma$. Similarly, $b$ (resp.~$c$) is the number of edges from one degree 3 vertex to another that are assigned 0 and 1 (resp. 1 and 0), in that order, by $\sigma$. We note that
 \begin{align*}
  w^a x^b y^c z^d + w^d x^c y^b z^a
  &= \left\{
  \begin{array}{ll}
   (w z)^{\min(a,d)} (x y)^{\min(b,c)} \left(w^{|a - d|} y^{|b - c|} + x^{|b - c|} z^{|a - d|}\right) & a > d \XOR b > c\\
   (w z)^{\min(a,d)} (x y)^{\min(b,c)} \left(w^{|a - d|} x^{|b - c|} + y^{|b - c|} z^{|a - d|}\right) & \text{otherwise}.
  \end{array}
  \right.
 \end{align*}

 We prove $a \equiv d \pmod{3}$ inductively. For the all-0 assignment, this is clear since every edge contributes a factor $w$ and the number of edges is divisible by 3 for a 3-regular graph. Now starting from any assignment $\sigma$, if we switch the assignment on one vertex from 0 to 1, it is easy to verify that it changes the valuation from $w^a x^b y^c z^d$ to $w^{a'} x^{b'} y^{c'} z^{d'}$, where $a - d = a' - d' + 3$. As every $\{0,1\}$ assignment is obtainable from the all-0 assignment by a sequence of switches, the conclusion $a \equiv d \pmod{3}$ follows.

 Now
 \begin{align*}
  w^a x^b y^c z^d + w^d x^c y^b z^a
  =
  \left\{
  \begin{array}{ll}
   (w z)^{\min(a,d)} (x y)^{\min(b,c)} \left(w^{3k} y^\ell + x^\ell z^{3k}\right) & a > d \XOR b > c\\
   (w z)^{\min(a,d)} (x y)^{\min(b,c)} \left(w^{3k} x^\ell + y^\ell z^{3k}\right) & \text{otherwise}
  \end{array}
  \right.
 \end{align*}
 for some $k, \ell \ge 0$. Consider $w^{3k} y^\ell + x^\ell z^{3k}$ (the other case is similar). Two simple inductive steps
 \begin{align*}
  w^{3k} y^{\ell + 1} + x^{\ell + 1} z^{3k} &= \left(w^{3k} y^\ell + x^\ell z^{3k}\right) (x + y) - x y \left(w^{3k} y^{\ell - 1} + x^{\ell - 1} z^{3k}\right)\\
  w^{3(k + 1)} y^\ell + x^\ell z^{3(k + 1)} &= \left(w^{3k} y^\ell + x^\ell z^{3k}\right) \left(w^3 + z^3\right) - (w z)^3 \left(w^{3(k - 1)} y^\ell + x^\ell z^{3(k - 1)}\right)
 \end{align*}
 (when combined with the other case) show that the Holant is a polynomial $P (w z, x y, w^3 + z^3, x + y, w^3 x + y z^3, w^3 y + x z^3)$ with integer coefficients.
\end{proof}

Assume non-degeneracy of $(w, x, y, z)$, Lemmas~\ref{lem:infOrder:xOverYNormNotOne}, \ref{lem:infOrder:x=0,wyz!=0}, and~\ref{lem:infOrder:w=0,xyz!=0} proved $\SHARPP$-hardness unless two (or more) of $w$, $x$, $y$, and $z$ are zero or none are zero and $|x| = |y|$. If any two (or more) of variables are zero, then the problem is tractable, as proved after Theorem~\ref{thm:dichotomy}. Therefore, the dichotomy in Theorem~\ref{thm:dichotomy} holds unless $w x y z \ne 0$ and $|x| = |y|$. In accordance with Lemma~\ref{lem:algebraicSymmetrization}, we make a change of variables to $A = w z$, $B = x y$, $C = w^3 + z^3$, $D = x + y$, $E = w^3 x + y z^3$, and $F = w^3 y + x z^3$. Since the complexity of a Holant remains the same under multiplication by a nonzero constant to any signature, we normalize so that $|x| = 1$ and $x = \overline{y}$ without repeatedly stating this as an assumption. Thus, $B = 1$ and $D = x + y \in [-2, 2]$ with $D^2 = 4$ corresponding to the symmetric case: $x = y$. A degenerate edge signature now means $A = 1$. Additionally, notice that $E + F = C D$ and $E F = -4 A^3 B + B C^2 + A^3 D^2$. Theorem~\ref{thm:dichotomy} can also be stated in these symmetrized variables.

\begin{theorem} \label{thm:dichotomySymmetrized}
 Suppose $w, x, y, z \in \C$. Then $\holant{(w,x,y,z)}{{=}_3}$ is $\SHARPP$-hard except in the following cases, for which the problem is in $\P$.

 $(1)\ w z = x y \iff A = B$.

 $(2)\ w = z = 0 \iff A = C = 0$.

 $(3)\ x = y = 0 \iff B = D = 0$.

 $(4)\ w z = - x y \myand w^6 = z^6 \myand x^2 = y^2 \iff A = -B \myand 4 A^3 C = C^3 \myand 4 B D = D^3$.

 $(5)\ w z = - x y \myand w^6 = -z^6 \myand x^2 = -y^2 \iff A = -B \myand 2 A^3 = C^2 \myand 2 B = D^2$.

 \noindent
 If the input is restricted to planar graphs, then two more cases become tractable but all other cases remain $\SHARPP$-hard.

 $(6)\ w^3 = z^3 \myand x = y \iff 4 A^3 = C^2 \myand 4 B = D^2$.

 $(7)\ w^3 = -z^3 \myand x = -y \iff C = D = 0$.
\end{theorem}

Now we continue with the proof of $\SHARPP$-hardness.

\begin{lemma} \label{lem:infOrder:aNotReal}
 If $D^2 \ne 4$, and $A \not\in \R$, then $\holant{(w,x,y,z)}{{=}_3}$ is $\SHARPP$-hard.
\end{lemma}
\begin{proof}
 The transition matrices for Gadgets~\ref{gadget:unary:0:001} and~\ref{gadget:unary:0:101} are $M_{\ref{gadget:unary:0:001}} = \shrinkMatrixTwoRows{\begin{bmatrix} w & y\\ x & z \end{bmatrix} \begin{bmatrix} w^2 & x y\\ x y & z^2 \end{bmatrix}}$ and $M_{\ref{gadget:unary:0:101}} = \shrinkMatrixTwoRows{\begin{bmatrix} w & x\\ y & z \end{bmatrix} \begin{bmatrix} w^2 & x y\\ x y & z^2 \end{bmatrix}}$. Both matrices have determinant $(A - 1)^2 (A + 1)$, which is nonzero since $A$ is not real. Then $N = M_{\ref{gadget:unary:0:001}} M_{\ref{gadget:unary:0:101}}^{-1}$ has determinant 1 and trace
 \begin{align*}
  \tr\left(\shrinkMatrixTwoRows{\begin{bmatrix} w & y\\ x & z \end{bmatrix} \begin{bmatrix} w & x\\ y & z \end{bmatrix}^{-1}}\right)
 = \frac{2 w z - x^2 - y^2}{w z - x y}
 = \frac{2 A - D^2 + 2}{A - 1},
 \end{align*}
 which is nonzero since $A$ is not real. If the eigenvalues of $N$ have distinct norms, then it has infinite order up to a scalar and we are done by Theorem~\ref{thm:interpolation}, so assume that its eigenvalues are of equal norm. Then Lemma~\ref{lem:2RootsSameNorm} says that $\frac{\tr(N)^2}{\det N} = \frac{(2 A - D^2 + 2)^2}{(A - 1)^2} \in \R^+$. Taking square roots, we have $\frac{2 A - D^2 + 2}{A - 1} \in \R$, which implies that $\frac{-D^2 + 4}{A - 1} \in \R$. Since $D^2 \ne 4$, this gives $A \in \R$, a contradiction.
\end{proof}

Unary recursive gadgets, such as the ones used in the proof of Lemma~\ref{lem:infOrder:aNotReal}, are quite useful for proving $\SHARPP$-hardness when variables like $A = w z$ are complex. When all variables are real, the conclusion of Lemma~\ref{lem:2RootsSameNorm} is weak (though one can still prove $\SHARPP$-hardness using a related lemma with significant effort in the symmetric case~\cite{CLX09b}). For complex variables in the symmetric case, \cite{KC10} showed that using higher arity (namely binary) recursive gadgets can give a much simpler proof of $\SHARPP$-hardness. The next lemma continues this pattern with the first ever use of ternary recursive gadgets.

\begin{lemma} \label{lem:infOrder:notHumanCheckable}
 If $A^2 \ne 1$, $A D \ne 0$, and $D^2 \ne 4$, then $\holant{(w,x,y,z)}{{=}_3}$ is $\SHARPP$-hard.
\end{lemma}
\begin{proof}
 The determinants of the 8-by-8 transition matrices of Gadget~\ref{gadget:ternary:0:000} and Gadget~\ref{gadget:ternary:1:001} are both $A^2 (A - 1)^4 \ne 0$. If $N = M_{\ref{gadget:ternary:0:000}}^{-1} M_{\ref{gadget:ternary:1:001}}$ has any two eigenvalues with distinct norms, then it has infinite order up to a scalar and we are done by Theorem~\ref{thm:interpolation}. Thus assume that all eight eigenvalues of $N$ have the same norm. Then by Lemma~\ref{lem:8RootsSameNorm}, we know that several equations hold among the coefficients of its characteristic polynomial. After scaling by the nonzero factor $A (A - 1)$, these coefficients for $A (A - 1) N$ are
 \begin{align*}
  a_7 &= (A - 1) (A D^2 + 2 A + 2)\\
  a_6 &= (A - 1)^2 (5 A^2 D^2 - 3 A^2 + 2 A D^2 + 2 A + 1)\\
  a_5 &= A (A - 1)^3 (A^2 D^4 + 5 A^2 D^2 - 6 A^2 + 7 A D^2 - 6 A + D^2)\\
  a_4 &= A^2 (A - 1)^4 (3 A^2 D^4 - 4 A^2 D^2 + 4 A^2 + A D^4 + 4 A D^2 - 4 A + 2 D^2 - 2)\\
  a_3 &= A^3 (A - 1)^5 (2 A D^4 + 3 A^2 D^4 - 6 A^2 D^2 + 6 A^2 - 4 A D^2 + 6 A + D^2)\\
  a_2 &= A^4 (A - 1)^6 (A^2 D^4 + A^2 D^2 - 3 A^2 + A D^4 - 2 A D^2 + 2 A + 1)\\
  a_1 &= A^6 (A - 1)^7 (2 A D^2 - 2 A + D^2 - 2)\\
  a_0 &= A^8 (A - 1)^8.
 \end{align*}
 Amazingly, $C$, $E$, and $F$ do not appear.\footnote{The runtime of \textsc{CylindricalDecomposition} is a double exponential in the number of variables, so it is crucial that our query include as few variables as possible.} Lemma~\ref{lem:infOrder:aNotReal} shows $\SHARPP$-hardness unless $A \in \R$, so assume that $A \in \R$. Because $A, D \in \R$, the equations in Lemma~\ref{lem:8RootsSameNorm} are simplified by the disappearance of norms and conjugates. Using \textsc{CylindricalDecomposition} in Mathematica\texttrademark, we conclude that there are no solutions under our assumptions, which is a contradiction.
\end{proof}

The meaning of the assumptions in Lemma~\ref{lem:infOrder:notHumanCheckable} will be explained after the next lemma, which considers the same assumptions except that $D$ is zero (a situation not covered in Lemma~\ref{lem:infOrder:notHumanCheckable}) and $C$ is nonzero.

\begin{lemma} \label{lem:infOrder:a^2!=1,ac!=0,d=0}
 If $A^2 \ne 1$, $A C \ne 0$, and $D = 0$, then $\holant{(w,x,y,z)}{{=}_3}$ is $\SHARPP$-hard.
\end{lemma}
\begin{proof}
 Lemma~\ref{lem:infOrder:aNotReal} shows $\SHARPP$-hardness unless $A \in \R$, so assume that $A \in \R$. The transition matrix for Gadget~\ref{gadget:binary:0:110} is $M_{\ref{gadget:binary:0:110}} = \shrinkMatrixTwoRows{\begin{bmatrix} w & x\\ y & z \end{bmatrix}^{\tensor 2}} \diag(w, x, y, z)$ and has determinant $A (A - 1)^2 \ne 0$. If $M_{\ref{gadget:binary:0:110}}$ has any two eigenvalues with distinct norms, then it has infinite order up to a scalar and we are done by Theorem~\ref{thm:interpolation}, so assume that all eigenvalues have the same norm. However, the coefficients of the characteristic polynomial of $M_{\ref{gadget:binary:0:110}}$, which are
 \begin{align*}
  (a_3, a_2, a_1, a_0) = \left(-C, (A + 1)^2 (A - 1), -(A - 1)^2 C, A (A - 1)^4\right),
 \end{align*}
 do not satisfy the conclusion of Lemma~\ref{lem:4RootsSameNorm} under the assumptions, a contradiction.
\end{proof}

The case $A = 1$ is degenerate (thus tractable), the case $A = 0$ is covered in Lemma~\ref{lem:infOrder:w=0,xyz!=0}, and recall that $D^2 = 4$ corresponds to the symmetric case~\cite{KC10}, so now we assume that $A \ne 0, 1 \myand D^2 \ne 4$. Lemma \ref{lem:infOrder:notHumanCheckable} handled $A \ne -1$ and $D \ne 0$ while Lemma \ref{lem:infOrder:a^2!=1,ac!=0,d=0} handled $A \ne -1 \myand D = 0 \myand C \ne 0$. We note that $C = D = 0$ is tractable on planar graphs. Now we focus on the case $A = -1$.

The next two proofs of $\SHARPP$-hardness (the proofs of Lemmas~\ref{lem:infOrder:a=-1,E^2!=0,-4} and~\ref{lem:infOrder:a=-1,F^2!=0,-4}) make use of the following technical lemma.

\begin{lemma} \label{lem:equSolsTrivial}
 Let $c \in \C$ and $\varepsilon = \pm 1$. Then the only solutions to the equation $(\overline{c + 2 \varepsilon}) c = \varepsilon (c + 2 \varepsilon)$ are the trivial solutions $c \in \{-2 \varepsilon, \varepsilon\}$.
\end{lemma}
\begin{proof}
 Assume that $c \ne -2 \varepsilon$. Now we show that $c = \varepsilon$. Taking norms, we see that $|c| = 1$. Then simplifying $(\overline{c + 2 \varepsilon}) c = \varepsilon (c + 2 \varepsilon)$ using $c \overline{c} = |c|^2 = 1$ yields $c = \varepsilon$ as claimed.
\end{proof}

\begin{lemma} \label{lem:infOrder:a=-1,E^2!=0,-4}
 If $A = -1$ and $E \not\in \{0, \pm 2 i\}$, then $\holant{(w,x,y,z)}{{=}_3}$ is $\SHARPP$-hard.
\end{lemma}
\begin{proof}
 The transition matrices of Gadgets~\ref{gadget:binary:4:110000}, \ref{gadget:binary:7:110010}, and~\ref{gadget:binary:5:110010} are
 \begin{align*}
  M_{\ref{gadget:binary:4:110000}}
  &=
   \begin{bmatrix}
    w & x\\
    y & z
   \end{bmatrix}^{\tensor 2}
   \diag(w, y, x, z)
   \begin{bmatrix}
    w & y\\
    x & z
   \end{bmatrix}^{\tensor 2}
   \diag(w, x, y, z)\\
  M_{\ref{gadget:binary:7:110010}}
  &=
   \begin{bmatrix}
    w & x\\
    y & z
   \end{bmatrix}^{\tensor 2}
   \begin{bmatrix}
    w^4 + w x y^2 & w^2 x y + x y^2 z & 0 & 0\\
    w^2 x y + x y^2 z & w x y z + y z^3 & 0 & 0\\
    0 & 0 & w^3 x + w x y z & w x^2 y + x y z^2\\
    0 & 0 & w x^2 y + x y z^2 & x^2 y z + z^4
   \end{bmatrix}\\
  M_{\ref{gadget:binary:5:110010}}
  &=
   \begin{bmatrix}
    w & x\\
    y & z
   \end{bmatrix}^{\tensor 2}
   \diag(w, y, x, z)
   \left(
   I_2 \tensor
   \begin{bmatrix}
    w & x\\
    y & z
   \end{bmatrix}
   \begin{bmatrix}
    w^2 + y z & 0\\
    0 & w x + z^2
   \end{bmatrix}
   \right)
 \end{align*}
 with $\det M_{\ref{gadget:binary:4:110000}} = 2^8$, $\det M_{\ref{gadget:binary:5:110010}} = -2^6 E^2$, and $\det M_{\ref{gadget:binary:7:110010}} = 2^6 (E^2 + 4)$, so all are nonsingular. Let $N_1 = M_{\ref{gadget:binary:4:110000}}^{-1} M_{\ref{gadget:binary:5:110010}}$ and $N_2 = M_{\ref{gadget:binary:4:110000}}^{-1} M_{\ref{gadget:binary:7:110010}}$. The coefficients of the characteristic polynomials of $-2^4 N_1$ and $2^4 N_2$ are respectively
 \begin{align*}
  (a_3, a_2, a_1, a_0) &= \left(-4, E^2 + 12, -4 (E^2 + 4), 4 (E^2 + 4)\right)\\
  (a_3, a_2, a_1, a_0) &= \left(4, -E^2 + 8, -4 E^2, -4 E^2\right).
 \end{align*}
 If $N_1$ (resp.~$N_2$) has any two eigenvalues with distinct norms, then $N_1$ (resp.~$N_2$) has infinite order up to a scalar and we are done by Theorem~\ref{thm:interpolation}, so assume that all eigenvalues of $N_1$ (resp.~$N_2$) have the same norm. Then by Lemma~\ref{lem:4RootsSameNorm}, we have two equations relating these coefficients. However, after a change of variables by $c = (E^2 + 4) / 4$ (for the coefficients of $N_1$) and $c = E^2 / 4$ (for the coefficients of $N_2$), Lemma~\ref{lem:equSolsTrivial} says that the only solutions to both equations require $E \in \{0, \pm 2 i\}$, a contradiction.
\end{proof}

The next lemma is similar to Lemma \ref{lem:infOrder:a=-1,E^2!=0,-4} with $E$ in place of $F$.

\begin{lemma} \label{lem:infOrder:a=-1,F^2!=0,-4}
 If $A = -1$ and $F \not\in \{0, \pm 2 i\}$, then $\holant{(w,x,y,z)}{{=}_3}$ is $\SHARPP$-hard.
\end{lemma}
\begin{proof}
 The transition matrices of Gadgets~\ref{gadget:binary:4:110000}, \ref{gadget:binary:7:111010}, and~\ref{gadget:binary:5:111100} are
 \begin{align*}
  M_{\ref{gadget:binary:4:110000}}
  &=
   \begin{bmatrix}
    w & x\\
    y & z
   \end{bmatrix}^{\tensor 2}
   \diag(w, y, x, z)
   \begin{bmatrix}
    w & x\\
    y & z
   \end{bmatrix}^{\tensor 2}
   \diag(w, x, y, z)\\
  M_{\ref{gadget:binary:7:111010}}
  &=
   \begin{bmatrix}
    w & x\\
    y & z
   \end{bmatrix}^{\tensor 2}
   \begin{bmatrix}
    w^4 + w x^2 y & w^2 x y + x^2 y z & 0 & 0\\
    w^2 x y + x^2 y z & w x y z + x z^3 & 0 & 0\\
    0 & 0 & w^3 y + w x y z & w x y^2 + x y z^2\\
    0 & 0 & w x y^2 + x y z^2 & x y^2 z + z^4
   \end{bmatrix}\\
  M_{\ref{gadget:binary:5:111100}}
  &=
   \begin{bmatrix}
    w & x\\
    y & z
   \end{bmatrix}^{\tensor 2}
   \diag(w, x, y, z)
   \left(
   I_2 \tensor
   \begin{bmatrix}
    w & x\\
    y & z
   \end{bmatrix}
   \begin{bmatrix}
    w^2 + x z & 0\\
    0 & w y + z^2
   \end{bmatrix}
   \right).
 \end{align*}
 The rest of the proof uses the same reasoning as the proof of Lemma~\ref{lem:infOrder:a=-1,E^2!=0,-4} with Gadgets~\ref{gadget:binary:7:110010} and~\ref{gadget:binary:5:110010} replaced by Gadgets~\ref{gadget:binary:7:111010} and~\ref{gadget:binary:5:111100} respectively.
\end{proof}

All remaining cases, those for which $A = -1$ and $E, F \in \{0, \pm 2 i\}$, imply tractability. Since this is not immediately obvious, we prove this next. As pointed out after Lemma~\ref{lem:algebraicSymmetrization}, the following equations hold and are used frequently below. They simplify to
\begin{align}
 E + F &= C D \label{equ:ePlusF}\\
 E F &= -4 A^3 B + B C^2 + A^3 D^2 = 4 + C^2 - D^2 \label{equ:eTimesF}
\end{align}
when $A = -1$ and $B = 1$. These next four lemmas cover all possibilities of $E, F \in \{0, \pm 2 i\}$ as follows:

Lemma \ref{lem:symmetrizedTractable:EF:BothZero}: Both zero

Lemma \ref{lem:symmetrizedTractable:EF:EqualNonzero}: Both nonzero and equal

Lemma \ref{lem:symmetrizedTractable:EF:NotEqualNonzero}: Both nonzero and not equal

Lemma \ref{lem:symmetrizedTractable:EF:OneZero}: Exactly one zero

\begin{lemma} \label{lem:symmetrizedTractable:EF:BothZero}
 If $A = -1 \myand E = F = 0$, then $(D = 0 \myand C^2 = -4)$ or $(D^2 = 4 \myand C = 0)$, which are both tractable.
\end{lemma}
\begin{proof}
 Since $0 = E + F = C D$, either $C$ or $D$ is zero. In either case, simplifying equation~(\ref{equ:eTimesF}) gives the desired result and is covered by tractable case~(4) in Theorem~\ref{thm:dichotomySymmetrized}.
\end{proof}

\begin{lemma} \label{lem:symmetrizedTractable:EF:EqualNonzero}
 If $A = -1 \myand E = F = \pm 2 i$, then $C^2 = -4 \myand D^2 = 4$, which is tractable.
\end{lemma}
\begin{proof}
 Using $w z = A = -1$ and $x y = B = 1$, we multiply $\pm 2 i = E = w^3 x + y z^3$ by $w^3 y$ to get $y^2 \pm 2 i w^3 y - w^6 = 0$. Similarly, multiplying $\pm 2 i = F = w^3 y + x z^3$ by $w^3 x$ gives $x^2 \pm 2 i w^3 x - w^6 = 0$. This is the same quadratic polynomial with $x$ and $y$ as indeterminates. Its discriminant is zero, so $x = y$ which means that $D^2 = 4$. Simplifying equation~(\ref{equ:eTimesF}) yields $C^2 = -4$ as required. This is covered by tractable case~(4) in Theorem~\ref{thm:dichotomySymmetrized}.
\end{proof}

\begin{lemma} \label{lem:symmetrizedTractable:EF:NotEqualNonzero}
 If $A = -1 \myand E = -F = \pm 2 i$, then $C = D = 0$, which is tractable.
\end{lemma}
\begin{proof}
 Since $0 = E + F = C D$, either $C$ or $D$ is zero. Simplifying equation~(\ref{equ:eTimesF}) gives $C^2 = D^2$, so both $C$ and $D$ are zero. This is covered by tractable case~(4) in Theorem~\ref{thm:dichotomySymmetrized}.
\end{proof}

\begin{lemma} \label{lem:symmetrizedTractable:EF:OneZero}
 If $A = -1 \myand ((E = \pm 2 i \myand F = 0) \myor (E = 0 \myand F = \pm 2 i))$, then $D^2 = 2 \myand C^2 = -2$, which is tractable.
\end{lemma}
\begin{proof}
 Since $\pm 2 i = E + F = C D$, neither $C$ or $D$ is zero. Squaring this equation and solving for $C^2$ gives $C^2 = -4 / D^2$. In equation~(\ref{equ:eTimesF}), first we substitute for $C^2$ to conclude that $D^2 = 2$ and then substitute for $D^2$ to conclude that $C^2 = -2$. This is tractable case~(5) in Theorem~\ref{thm:dichotomySymmetrized}.
\end{proof}

\begin{figure}[t]
 \centering
 \captionsetup[subfigure]{labelformat=empty}
 \gadgetSubfloat[Gadget \ref{gadget:symmetricGenerator1}]{
  \begin{tikzpicture}[scale=\scale,transform shape,>=\arrowType,node distance=\nodeDist,semithick]
   \node[external] (0)              {};
   \node[internal] (2) [right of=0] {};
   \node[internal] (1) [right of=2] {};
   \node[external] (3) [right of=1] {};
   \path (0) edge [->]             node[pos=\startPos] (e1) {} (2)
         (1) edge [<-, bend right] node[midway]        (e2) {} (2)
             edge [->, bend left]  node[midway]        (e3) {} (2)
             edge [<-]             node[pos=\endPos]   (e4) {} (3);
   \begin{pgfonlayer}{background}
    \node[draw=\borderColor,thick,rounded corners,fit = (1) (2) (e1) (e2) (e3) (e4)] {};
   \end{pgfonlayer}
  \end{tikzpicture}} \label{gadget:symmetricGenerator1}
 \gadgetSubfloat[Gadget \ref{gadget:symmetricGenerator2}]{
  \begin{tikzpicture}[scale=\scale,transform shape,>=\arrowType,node distance=\nodeDist,semithick]
   \node[external] (0)              {};
   \node[internal] (2) [right of=0] {};
   \node[internal] (1) [right of=2] {};
   \node[external] (3) [right of=1] {};
   \path (0) edge [<-]             node[pos=\startPos] (e1) {} (2)
         (1) edge [<-, bend right] node[midway]        (e2) {} (2)
             edge [->, bend left]  node[midway]        (e3) {} (2)
             edge [->]             node[pos=\endPos]   (e4) {} (3);
   \begin{pgfonlayer}{background}
    \node[draw=\borderColor,thick,rounded corners,fit = (1) (2) (e1) (e2) (e3) (e4)] {};
   \end{pgfonlayer}
  \end{tikzpicture}} \label{gadget:symmetricGenerator2}
 \caption{Gadgets with a symmetric generator signature}
\end{figure}
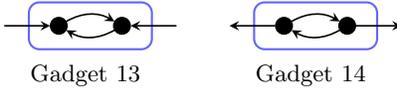

At this point, every setting of the variables has either been proven tractable over planar graphs or $\SHARPP$-hard. So far, all our hardness proofs originate from \#\textsc{VertexCover} on 3-regular graphs, which is $\holant{(0,1,1,1)}{{=_3}}$ (see the proof of Lemma~\ref{lem:hardByUnary} in section~\ref{sec:hardByUnaryProof}). Recall that \#\textsc{VertexCover} is $\SHARPP$-hard even for 3-regular planar graphs~\cite{XZZ07} and notice that all of our gadget constructions are planar, including our interpolation construction in the Group Lemma (see Figure~\ref{gadget:example:interpolationConstruction}). Therefore, all of the $\SHARPP$-hardness results proved so far still apply when the input is restricted to planar graphs. There are, however, some cases where the problem is $\SHARPP$-hard in general, yet is polynomial time computable when restricted to planar graphs. We analyze this case next using a lemma from~\cite{KC10} that can also be found in~\cite{MikeThesis}.
\begin{lemma}[Lemma 33 of \cite{MikeThesis}] \label{lem:planar:SymmetircHard}
 The problem $\holant{(w,1,1,w)}{{=}_3}$ is $\SHARPP$-hard unless $w \in \{0, \pm 1, \pm i\}$, in which case it is in $\P$.
\end{lemma}

\begin{lemma} \label{lem:planar:AsymmetricHard}
 The problem $\holant{(w,1,-1,-w)}{{=}_3}$ is $\SHARPP$-hard, unless $w \in \{0, \pm 1, \pm i\}$, in which case it is in $\P$.
\end{lemma}
\begin{proof}
 If $w \in \{0, \pm 1, \pm i\}$, then the problem is in $\P$ by the tractability proof of Theorem~\ref{thm:dichotomy}. If $w \not\in \{0, \pm 1, \pm i\}$, then Gadgets~\ref{gadget:symmetricGenerator1} and~\ref{gadget:symmetricGenerator2} simulate two symmetric generators. Using \textsc{Reduce} in Mathematica\texttrademark, we conclude that at least one of the gadgets satisfies the hypothesis for $\SHARPP$-hardness from Lemma~\ref{lem:planar:SymmetircHard}.
\end{proof}

\begin{lemma}
 If $w^3 = \varepsilon z^3 \myand x = \varepsilon y$ where $\varepsilon = \pm 1$, then $\holant{(w,x,y,z)}{{=}_3}$ is $\SHARPP$-hard unless $x = 0 \myor w / x \in \{0, \pm 1, \pm i\}$, in which case it is in $\P$.
\end{lemma}
\begin{proof}
 If $x = 0$, then also $y = 0$ and this case, generalized equality, is tractable in Theorem~\ref{thm:dichotomy}. Now assume $x \ne 0$. If $w z \ne 0$, then we apply the holographic transformation $\shrinkMatrixTwoRows{\begin{bmatrix} \alpha & 0\\ 0 & \alpha^2 \end{bmatrix}}$ with $\alpha = \varepsilon z / w$. Note that $\alpha^3 = \varepsilon w^3 / z^3 = 1$. The edge signature becomes $(w,x,y,z) \shrinkMatrixTwoRows{\begin{bmatrix} \alpha & 0\\ 0 & \alpha^2 \end{bmatrix}^{\tensor 2}} = (\alpha^2 w, x, y, \alpha z)$, while $=_3$ is unchanged since $\shrinkMatrixTwoRows{\left(\begin{bmatrix} \alpha & 0\\ 0 & \alpha^2 \end{bmatrix}^{-1}\right)^{\tensor 3}} = I$, the 8-by-8 identity matrix~\cite{Val08,CL11}. This reduces to the case $w = \varepsilon z \myand x = \varepsilon y$. We note that when $w z = 0$, this equivalence still holds. We then normalize $x = 1$ (since it is nonzero) and replace $z$ with $\varepsilon w$ to obtain the edge signature $(w / x, 1, \varepsilon, \varepsilon w / x)$. Depending on $\varepsilon$, this case is either covered in Lemma~\ref{lem:planar:SymmetircHard} or Lemma~\ref{lem:planar:AsymmetricHard}, so we are done.
\end{proof}

\section{Tractable Signatures and Quadratic Polynomials} \label{sec:affine}
Here we briefly discuss an alternative and more conceptual catalog of all the tractable functions $f$. In Theorem~\ref{thm:dichotomy}, Case~(1) is degenerate: $\det \begin{bmatrix} f(0,0) & f(0,1) \\ f(1,0) & f(1,1) \end{bmatrix} = 0$. Case~(2) is generalized disequality: $\begin{bmatrix} 0 & f(0,1) \\ f(1,0) & 0 \end{bmatrix}$. Case~(3) is generalized equality: $\begin{bmatrix} f(0,0) & 0 \\ 0 & f(1,1) \end{bmatrix}$. The tractable Case~(4) of Theorem~\ref{thm:dichotomy} is more interesting: for $f = \begin{bmatrix} w & x \\ y & z \end{bmatrix}$, we have $w z = - x y \myand w^6 = \varepsilon z^6 \myand x^2 = \varepsilon y^2$, where $\varepsilon = \pm 1$. If any of $w,x,y,z = 0$, then the above condition forces all $w = x = y = z = 0$. This constant 0 function is trivially tractable. Assume $w x y z \not = 0$. In the proof of tractability of
Theorem~\ref{thm:dichotomy}, it is shown that under a holographic transformation, this is equivalent to $w z = - x y \myand w^2 = \varepsilon z^2 \myand x^2 = \varepsilon y^2$, where $\varepsilon = \pm 1$. If we further normalize $f$ by setting $w = 1$, which corresponds to a global nonzero constant factor, then this tractable case is equivalent to $z = - x y \myand 1 = \varepsilon z^2 \myand x^2 = \varepsilon y^2$, where $\varepsilon = \pm 1$. This has a simple form as an exponential quadratic polynomial.

\begin{lemma}
 For $x_1, x_2 \in \{0,1\}$, let $f(x_1, x_2) = (1, x, y, z)$. Then $z = - x y \myand \varepsilon = z^2 \myand x^2 = \varepsilon y^2$ with $\varepsilon = \pm 1$ iff there exists $a, b \in \Z / 4 \Z$ such that
\[f(x_1, x_2) = i^{2 x_1 x_2 + b x_1 + c x_2}.\]
\end{lemma}
\begin{proof}
 The proof is by a direct verification on the 16 possible cases, eight of which have $\varepsilon = 1$ which correspond to $a + b \equiv 0 \pmod{2}$ while the other eight have $\varepsilon = -1$ which correspond to $a + b \equiv 1 \pmod{2}$.
\end{proof}

Note that without normalizing $w = 1$, and including some degenerate cases, we can use \[f(x_1, x_2) = \lambda i^{2 a x_1 x_2 + b x_1 + c x_2},\] where $\lambda \in \C$ only contributes a constant factor, and the case $a = 0$ is degenerate.

\pagebreak
\section{Gadgets}
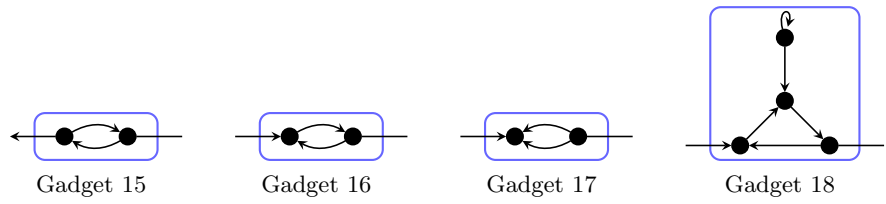
\begin{figure}[ht]
 \centering
 \captionsetup[subfigure]{labelformat=empty}
 \gadgetSubfloat[Gadget \ref{gadget:unary:0:001}]{
  \begin{tikzpicture}[scale=\scale,transform shape,>=\arrowType,node distance=\nodeDist,semithick]
   \node[external] (0)              {};
   \node[internal] (2) [right of=0] {};
   \node[internal] (1) [right of=2] {};
   \node[external] (3) [right of=1] {};
   \path (0) edge [<-]             node[pos=\startPos] (e1) {} (2)
         (1) edge [<-, bend right] node[midway]        (e2) {} (2)
             edge [->, bend left]  node[midway]        (e3) {} (2)
             edge [-]              node[pos=\endPos]   (e4) {} (3);
   \begin{pgfonlayer}{background}
    \node[draw=\borderColor,thick,rounded corners,fit = (1) (2) (e1) (e2) (e3) (e4)] {};
   \end{pgfonlayer}
  \end{tikzpicture}} \label{gadget:unary:0:001}
 \gadgetSubfloat[Gadget \ref{gadget:unary:0:101}]{
  \begin{tikzpicture}[scale=\scale,transform shape,>=\arrowType,node distance=\nodeDist,semithick]
   \node[external] (0)              {};
   \node[internal] (2) [right of=0] {};
   \node[internal] (1) [right of=2] {};
   \node[external] (3) [right of=1] {};
   \path (0) edge [->]             node[pos=\startPos] (e1) {} (2)
         (1) edge [<-, bend right] node[midway]        (e2) {} (2)
             edge [->, bend left]  node[midway]        (e3) {} (2)
             edge [-]              node[pos=\endPos]   (e4) {} (3);
   \begin{pgfonlayer}{background}
    \node[draw=\borderColor,thick,rounded corners,fit = (1) (2) (e1) (e2) (e3) (e4)] {};
   \end{pgfonlayer}
  \end{tikzpicture}} \label{gadget:unary:0:101}
 \gadgetSubfloat[Gadget \ref{gadget:unary:0:111}]{
  \begin{tikzpicture}[scale=\scale,transform shape,>=\arrowType,node distance=\nodeDist,semithick]
   \node[external] (0)              {};
   \node[internal] (2) [right of=0] {};
   \node[internal] (1) [right of=2] {};
   \node[external] (3) [right of=1] {};
   \path (0) edge [->]             node[pos=\startPos] (e1) {} (2)
         (1) edge [->, bend right] node[midway]        (e2) {} (2)
             edge [->, bend left]  node[midway]        (e3) {} (2)
             edge [-]              node[pos=\endPos]   (e4) {} (3);
   \begin{pgfonlayer}{background}
    \node[draw=\borderColor,thick,rounded corners,fit = (1) (2) (e1) (e2) (e3) (e4)] {};
   \end{pgfonlayer}
  \end{tikzpicture}} \label{gadget:unary:0:111}
 \gadgetSubfloat[Gadget \ref{gadget:unary:4:101010}]{
  \begin{tikzpicture}[scale=\scale,transform shape,>=\arrowType,node distance=\nodeDist,semithick]
   \node[external] (0)                    {};
   \node[internal] (4) [right of=0]       {};
   \node[internal] (3) [above right of=4] {};
   \node[internal] (1) [below right of=3] {};
   \node[internal] (2) [above of=3]       {};
   \node[external] (5) [right of=1]       {};
   \path (0) edge [->]             node[pos=\startPos] (e1) {} (4)
         (1) edge [<-]                                         (3)
             edge [->]                                         (4)
         (2) edge [<-, loop above] coordinate          (c1)    (2)
             edge [->]                                         (3)
         (3) edge [<-]                                         (4)
         (1) edge [-]              node[pos=\endPos]   (e2) {} (5);
   \begin{pgfonlayer}{background}
    \node[draw=\borderColor,thick,rounded corners,fit = (1) (2) (3) (4) (e1) (e2) (c1)] {};
   \end{pgfonlayer}
  \end{tikzpicture}} \label{gadget:unary:4:101010}
 \caption{Unary recursive gadgets}
\end{figure}

\begin{figure}[ht]
 \centering
 \captionsetup[subfigure]{labelformat=empty}
 \gadgetSubfloat[Gadget \ref{gadget:binary:0:110}]{
  \begin{tikzpicture}[scale=\scale,transform shape,>=\arrowType,node distance=\nodeDist,semithick]
   \node[external] (0)              {};
   \node[external] (1) [below of=0] {};
   \node[internal] (2) [right of=1] {};
   \node[internal] (3) [right of=0] {};
   \node[external] (4) [right of=2] {};
   \node[external] (5) [right of=3] {};
   \path (0) edge [->] node[pos=\startPos] (e1) {} (3)
         (1) edge [->] node[pos=\startPos] (e2) {} (2)
         (2) edge [<-]                             (3)
             edge [-]  node[pos=\endPos]   (e3) {} (4)
         (3) edge [-]  node[pos=\endPos]   (e4) {} (5);
   \begin{pgfonlayer}{background}
    \node[draw=\borderColor,thick,rounded corners,fit = (2) (3) (e1) (e2) (e3) (e4)] {};
   \end{pgfonlayer}
  \end{tikzpicture}} \label{gadget:binary:0:110}
 \gadgetSubfloat[Gadget \ref{gadget:binary:0:111}]{
  \begin{tikzpicture}[scale=\scale,transform shape,>=\arrowType,node distance=\nodeDist,semithick]
   \node[external] (0)              {};
   \node[external] (1) [below of=0] {};
   \node[internal] (2) [right of=1] {};
   \node[internal] (3) [right of=0] {};
   \node[external] (4) [right of=2] {};
   \node[external] (5) [right of=3] {};
   \path (0) edge [->] node[pos=\startPos] (e1) {} (3)
         (1) edge [->] node[pos=\startPos] (e2) {} (2)
         (2) edge [->]                             (3)
             edge [-]  node[pos=\endPos]   (e3) {} (4)
         (3) edge [-]  node[pos=\endPos]   (e4) {} (5);
   \begin{pgfonlayer}{background}
    \node[draw=\borderColor,thick,rounded corners,fit = (2) (3) (e1) (e2) (e3) (e4)] {};
   \end{pgfonlayer}
  \end{tikzpicture}} \label{gadget:binary:0:111}\\
 \gadgetSubfloat[Gadget \ref{gadget:binary:4:110000}]{
  \begin{tikzpicture}[scale=\scale,transform shape,>=\arrowType,node distance=\nodeDist,semithick]
   \node[external] (0)              {};
   \node[external] (1) [below of=0] {};
   \node[internal] (4) [right of=1] {};
   \node[internal] (5) [right of=0] {};
   \node[internal] (2) [right of=5] {};
   \node[internal] (3) [right of=4] {};
   \node[external] (6) [right of=2] {};
   \node[external] (7) [right of=3] {};
   \path (0) edge [->] node[pos=\startPos] (e1) {} (5)
         (1) edge [->] node[pos=\startPos] (e2) {} (4)
         (2) edge [<-]                             (3)
             edge [<-]                             (5)
         (3) edge [<-]                             (4)
         (4) edge [<-]                             (5)
         (2) edge [-]  node[pos=\endPos]   (e3) {} (6)
         (3) edge [-]  node[pos=\endPos]   (e4) {} (7);
   \begin{pgfonlayer}{background}
    \node[draw=\borderColor,thick,rounded corners,fit = (2) (3) (4) (5) (e1) (e2) (e3) (e4)] {};
   \end{pgfonlayer}
  \end{tikzpicture}} \label{gadget:binary:4:110000}
 \gadgetSubfloat[Gadget \ref{gadget:binary:7:110010}]{
  \begin{tikzpicture}[scale=\scale,transform shape,>=\arrowType,node distance=\nodeDist,semithick]
   \node[external] (1) [below of=0] {};
   \node[internal] (5) [right of=1] {};
   \node[internal] (4) [above right of=5] {};
   \node[internal] (2) [above of=4] {};
   \node[internal] (3) [below right of=4] {};
   \node[external] (7) [right of=3] {};
   \path let
          \p1 = (1),
          \p2 = (2)
         in
          node[external] (0) at (\x1,\y2) {};
   \path let
          \p1 = (7),
          \p2 = (2)
         in
          node[external] (6) at (\x1,\y2) {};
   \path (0) edge [->] node[pos=\startPos] (e1) {} (2)
         (1) edge [->] node[pos=\startPos] (e2) {} (5)
         (2) edge [<-]                             (4)
         (3) edge [<-]                             (4)
             edge [->]                             (5)
         (4) edge [<-]                             (5)
         (2) edge [-]  node[pos=\endPos]   (e3) {} (6)
         (3) edge [-]  node[pos=\endPos]   (e4) {} (7);
   \begin{pgfonlayer}{background}
    \node[draw=\borderColor,thick,rounded corners,fit = (2) (3) (4) (5) (e1) (e2) (e3) (e4)] {};
   \end{pgfonlayer}
  \end{tikzpicture}} \label{gadget:binary:7:110010}
 \gadgetSubfloat[Gadget \ref{gadget:binary:7:111010}]{
  \begin{tikzpicture}[scale=\scale,transform shape,>=\arrowType,node distance=\nodeDist,semithick]
   \node[external] (1) [below of=0] {};
   \node[internal] (5) [right of=1] {};
   \node[internal] (4) [above right of=5] {};
   \node[internal] (2) [above of=4] {};
   \node[internal] (3) [below right of=4] {};
   \node[external] (7) [right of=3] {};
   \path let
          \p1 = (1),
          \p2 = (2)
         in
          node[external] (0) at (\x1,\y2) {};
   \path let
          \p1 = (7),
          \p2 = (2)
         in
          node[external] (6) at (\x1,\y2) {};
   \path (0) edge [->] node[pos=\startPos] (e1) {} (2)
         (1) edge [->] node[pos=\startPos] (e2) {} (5)
         (2) edge [->]                             (4)
         (3) edge [<-]                             (4)
             edge [->]                             (5)
         (4) edge [<-]                             (5)
         (2) edge [-]  node[pos=\endPos]   (e3) {} (6)
         (3) edge [-]  node[pos=\endPos]   (e4) {} (7);
   \begin{pgfonlayer}{background}
    \node[draw=\borderColor,thick,rounded corners,fit = (2) (3) (4) (5) (e1) (e2) (e3) (e4)] {};
   \end{pgfonlayer}
  \end{tikzpicture}} \label{gadget:binary:7:111010}
 \gadgetSubfloat[Gadget \ref{gadget:binary:5:110010}]{
  \begin{tikzpicture}[scale=\scale,transform shape,>=\arrowType,node distance=\nodeDist,semithick]
   \node[external] (0)              {};
   \node[external] (1) [below of=0] {};
   \node[internal] (2) [right of=0] {};
   \node[internal] (5) [right of=1] {};
   \node[internal] (3) [right of=5] {};
   \node[internal] (4) [below of=3] {};
   \node[external] (7) [right of=3] {};
   \node[external] (6) [above of=7] {};
   \path (0) edge [->]             node[pos=\startPos] (e1) {} (2)
         (1) edge [->]             node[pos=\startPos] (e2) {} (5)
         (2) edge [<-]                                         (5)
         (3) edge [<-]                                         (4)
             edge [->]                                         (5)
         (4) edge [<-, loop below] node[pos=\startPos] (e1) {} (4)
         (2) edge [-]              node[pos=\endPos]   (e3) {} (6)
         (3) edge [-]              node[pos=\endPos]   (e4) {} (7);
   \begin{pgfonlayer}{background}
    \node[draw=\borderColor,thick,rounded corners,fit = (2) (3) (4) (5) (e1) (e2) (e3) (e4) (e1)] {};
   \end{pgfonlayer}
  \end{tikzpicture}} \label{gadget:binary:5:110010}
 \gadgetSubfloat[Gadget \ref{gadget:binary:5:111100}]{
  \begin{tikzpicture}[scale=\scale,transform shape,>=\arrowType,node distance=\nodeDist,semithick]
   \node[external] (0)              {};
   \node[external] (1) [below of=0] {};
   \node[internal] (2) [right of=0] {};
   \node[internal] (5) [right of=1] {};
   \node[internal] (3) [right of=5] {};
   \node[internal] (4) [below of=3] {};
   \node[external] (7) [right of=3] {};
   \node[external] (6) [above of=7] {};
   \path (0) edge [->]             node[pos=\startPos] (e1) {} (2)
         (1) edge [->]             node[pos=\startPos] (e2) {} (5)
         (2) edge [->]                                         (5)
         (3) edge [->]                                         (4)
             edge [<-]                                         (5)
         (4) edge [<-, loop below] node[pos=\startPos] (e1) {} (4)
         (2) edge [-]              node[pos=\endPos]   (e3) {} (6)
         (3) edge [-]              node[pos=\endPos]   (e4) {} (7);
   \begin{pgfonlayer}{background}
    \node[draw=\borderColor,thick,rounded corners,fit = (2) (3) (4) (5) (e1) (e2) (e3) (e4) (e1)] {};
   \end{pgfonlayer}
  \end{tikzpicture}} \label{gadget:binary:5:111100}
 \caption{Binary recursive gadgets}
\end{figure}
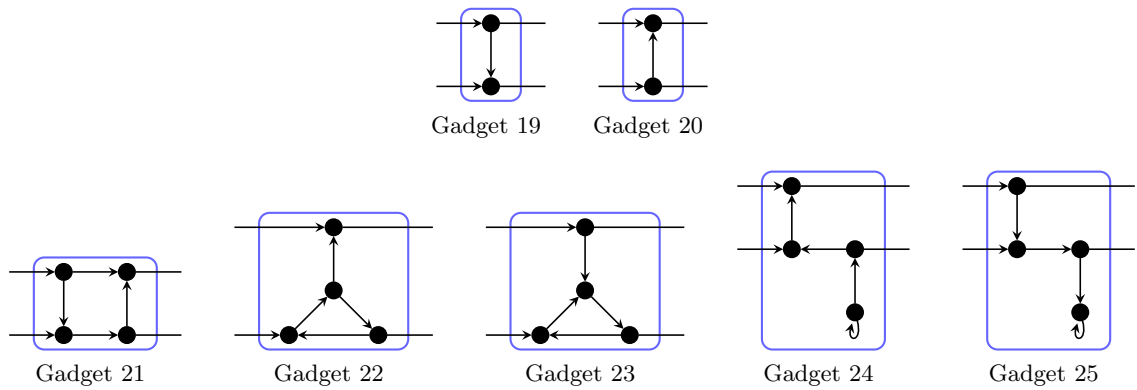

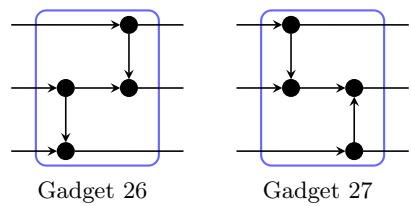
\begin{figure}[ht]
 \centering
 \captionsetup[subfigure]{labelformat=empty}
 \gadgetSubfloat[Gadget \ref{gadget:ternary:0:000}]{
  \begin{tikzpicture}[scale=\scale,transform shape,>=\arrowType,node distance=\nodeDist,semithick]
   \node[external] (0)              {};
   \node[external] (1) [below of=0] {};
   \node[external] (2) [below of=1] {};
   \node[internal] (5) [right of=2] {};
   \node[internal] (6) [above of=5] {};
   \node[internal] (4) [right of=6] {};
   \node[internal] (3) [above of=4] {};
   \node[external] (7) [right of=3] {};
   \node[external] (8) [below of=7] {};
   \node[external] (9) [below of=8] {};
   \path (0) edge [->] node[pos=\startPos] (e1) {} (3)
         (1) edge [->] node[pos=\startPos] (e2) {} (6)
         (2) edge [->] node[pos=\startPos] (e3) {} (5)
         (6) edge [->]                             (5)
             edge [->]                             (4)
         (3) edge [->]                             (4)
         (3) edge [-]  node[pos=\endPos]   (e4) {} (7)
         (4) edge [-]  node[pos=\endPos]   (e5) {} (8)
         (5) edge [-]  node[pos=\endPos]   (e6) {} (9);
   \begin{pgfonlayer}{background}
    \node[draw=\borderColor,thick,rounded corners,fit = (3) (5) (e1) (e2) (e3) (e4) (e5) (e6)] {};
   \end{pgfonlayer}
  \end{tikzpicture}} \label{gadget:ternary:0:000}
 \gadgetSubfloat[Gadget \ref{gadget:ternary:1:001}]{
  \begin{tikzpicture}[scale=\scale,transform shape,>=\arrowType,node distance=\nodeDist,semithick]
   \node[external] (0)              {};
   \node[external] (1) [below of=0] {};
   \node[external] (2) [below of=1] {};
   \node[internal] (3) [right of=0] {};
   \node[internal] (6) [below of=3] {};
   \node[internal] (4) [right of=6] {};
   \node[internal] (5) [below of=4] {};
   \node[external] (9) [right of=5] {};
   \node[external] (8) [above of=9] {};
   \node[external] (7) [above of=8] {};
   \path (0) edge [->] node[pos=\startPos] (e1) {} (3)
         (1) edge [->] node[pos=\startPos] (e2) {} (6)
         (2) edge [->] node[pos=\startPos] (e3) {} (5)
         (3) edge [->]                             (6)
         (6) edge [->]                             (4)
         (5) edge [->]                             (4)
         (3) edge [-]  node[pos=\endPos]   (e4) {} (7)
         (4) edge [-]  node[pos=\endPos]   (e5) {} (8)
         (5) edge [-]  node[pos=\endPos]   (e6) {} (9);
   \begin{pgfonlayer}{background}
    \node[draw=\borderColor,thick,rounded corners,fit = (3) (5) (e1) (e2) (e3) (e4) (e5) (e6)] {};
   \end{pgfonlayer}
  \end{tikzpicture}} \label{gadget:ternary:1:001}
 \caption{Ternary recursive gadgets}
\end{figure}

\begin{figure}[ht]
 \centering
 \captionsetup[subfigure]{labelformat=empty}
 \gadgetSubfloat[Gadget \ref{gadget:finish:0:1}]{
  \begin{tikzpicture}[scale=\scale,transform shape,>=\arrowType,node distance=\nodeDist,semithick]
   \node[external] (0)              {};
   \node[internal] (1) [right of=0] {};
   \node[external] (2) [above right of=1,yshift=-13] {};
   \node[external] (3) [below right of=1,yshift=13]  {};
   \path (0) edge [->]                 node[pos=\startPos] (e1) {} (1)
         (1) edge [-, out=90, in=180]  node[pos=0.4]       (e2) {} (2)
             edge [-, out=-90, in=180] node[pos=0.4]       (e3) {} (3);
   \begin{pgfonlayer}{background}
    \node[draw=\borderColor,thick,rounded corners,fit = (1) (e1) (e2) (e3)] {};
   \end{pgfonlayer}
  \end{tikzpicture}} \label{gadget:finish:0:1}
 \gadgetSubfloat[Gadget \ref{gadget:finish:2:1000}]{
  \begin{tikzpicture}[scale=\scale,transform shape,>=\arrowType,node distance=\nodeDist,semithick]
   \node[external] (0)              {};
   \node[internal] (3) [right of=0] {};
   \node[internal] (1) [above right of=3] {};
   \node[internal] (2) [below right of=3] {};
   \node[external] (4) [right of=1] {};
   \node[external] (5) [right of=2] {};
   \path (0) edge [->] node[pos=\startPos] (e1) {} (3)
         (1) edge [<-]                             (2)
             edge [<-]                             (3)
         (2) edge [<-]                             (3)
         (1) edge [-]  node[pos=\endPos]   (e2) {} (4)
         (2) edge [-]  node[pos=\endPos]   (e3) {} (5);
   \begin{pgfonlayer}{background}
    \node[draw=\borderColor,thick,rounded corners,fit = (1) (2) (3) (e1) (e2) (e3)] {};
   \end{pgfonlayer}
  \end{tikzpicture}} \label{gadget:finish:2:1000}
 \gadgetSubfloat[Gadget \ref{gadget:finish:2:1010}]{
  \begin{tikzpicture}[scale=\scale,transform shape,>=\arrowType,node distance=\nodeDist,semithick]
   \node[external] (0)              {};
   \node[internal] (3) [right of=0] {};
   \node[internal] (1) [above right of=3] {};
   \node[internal] (2) [below right of=3] {};
   \node[external] (4) [right of=1] {};
   \node[external] (5) [right of=2] {};
   \path (0) edge [->] node[pos=\startPos] (e1) {} (3)
         (1) edge [<-]                             (2)
             edge [->]                             (3)
         (2) edge [<-]                             (3)
         (1) edge [-]  node[pos=\endPos]   (e2) {} (4)
         (2) edge [-]  node[pos=\endPos]   (e3) {} (5);
   \begin{pgfonlayer}{background}
    \node[draw=\borderColor,thick,rounded corners,fit = (1) (2) (3) (e1) (e2) (e3)] {};
   \end{pgfonlayer}
  \end{tikzpicture}} \label{gadget:finish:2:1010}
 \gadgetSubfloat[Gadget \ref{gadget:finish:2:1011}]{
  \begin{tikzpicture}[scale=\scale,transform shape,>=\arrowType,node distance=\nodeDist,semithick]
   \node[external] (0)              {};
   \node[internal] (3) [right of=0] {};
   \node[internal] (1) [above right of=3] {};
   \node[internal] (2) [below right of=3] {};
   \node[external] (4) [right of=1] {};
   \node[external] (5) [right of=2] {};
   \path (0) edge [->] node[pos=\startPos] (e1) {} (3)
         (1) edge [<-]                             (2)
             edge [->]                             (3)
         (2) edge [->]                             (3)
         (1) edge [-]  node[pos=\endPos]   (e2) {} (4)
         (2) edge [-]  node[pos=\endPos]   (e3) {} (5);
   \begin{pgfonlayer}{background}
    \node[draw=\borderColor,thick,rounded corners,fit = (1) (2) (3) (e1) (e2) (e3)] {};
   \end{pgfonlayer}
  \end{tikzpicture}} \label{gadget:finish:2:1011}
 \gadgetSubfloat[Gadget \ref{gadget:finish:3:1000}]{
  \begin{tikzpicture}[scale=\scale,transform shape,>=\arrowType,node distance=\nodeDist,semithick]
   \node[external] (0)              {};
   \node[internal] (1) [right of=0] {};
   \node[internal] (3) [below right of=1] {};
   \node[internal] (2) [right of=3] {};
   \node[external] (5) [right of=2] {};
   \node[external] (a) [above right of=1] {};
   \node[external] (b) [right of=a] {};
   \node[external] (4) [right of=b] {};
   \path (0) edge [->]                               node[pos=\startPos] (e1) {} (1)
         (1) edge [<-, bend right]                                               (3)
         (2) edge [<-, bend left]                    node[midway]        (e4) {} (3)
             edge [<-, bend right]                                               (3)
         (1) edge [-, out=90, in=180, looseness=0.6] node[pos=\endPos]   (e2) {} (4)
         (2) edge [-]                                node[pos=\endPos]   (e3) {} (5);
   \begin{pgfonlayer}{background}
    \node[draw=\borderColor,thick,rounded corners,fit = (1) (2) (3) (a) (b) (e1) (e2) (e3) (e4)] {};
   \end{pgfonlayer}
  \end{tikzpicture}} \label{gadget:finish:3:1000}\\
 \gadgetSubfloat[Gadget \ref{gadget:finish:3:1010}]{
  \begin{tikzpicture}[scale=\scale,transform shape,>=\arrowType,node distance=\nodeDist,semithick]
   \node[external] (0)              {};
   \node[internal] (1) [right of=0] {};
   \node[internal] (3) [below right of=1] {};
   \node[internal] (2) [right of=3] {};
   \node[external] (5) [right of=2] {};
   \node[external] (a) [above right of=1] {};
   \node[external] (b) [right of=a] {};
   \node[external] (4) [right of=b] {};
   \path (0) edge [->]                               node[pos=\startPos] (e1) {} (1)
         (1) edge [<-, bend right]                                               (3)
         (2) edge [->, bend left]                    node[midway]        (e4) {} (3)
             edge [<-, bend right]                                               (3)
         (1) edge [-, out=90, in=180, looseness=0.6] node[pos=\endPos]   (e2) {} (4)
         (2) edge [-]                                node[pos=\endPos]   (e3) {} (5);
   \begin{pgfonlayer}{background}
    \node[draw=\borderColor,thick,rounded corners,fit = (1) (2) (3) (a) (b) (e1) (e2) (e3) (e4)] {};
   \end{pgfonlayer}
  \end{tikzpicture}} \label{gadget:finish:3:1010}
 \gadgetSubfloat[Gadget \ref{gadget:finish:14:1001010}]{
  \begin{tikzpicture}[scale=\scale,transform shape,>=\arrowType,node distance=\nodeDist,semithick]
   \node[external] (0)              {};
   \node[internal] (5) [right of=0] {};
   \node[internal] (3) [below right of=5] {};
   \node[internal] (4) [above right of=3] {};
   \node[internal] (2) [below right of=4] {};
   \node[internal] (1) [above of=4] {};
   \node[external] (7) [right of=2] {};
   \path let
          \p1 = (7),
          \p2 = (1)
         in
          node[external] (6) at (\x1,\y2) {};
   \path (0) edge [->]                 node[pos=\startPos] (e1) {} (5)
         (1) edge [<-]                                             (4)
             edge [<-, out=180, in=90]                             (5)
         (2) edge [->]                                             (3)
             edge [<-]                                             (4)
         (3) edge [->]                                             (4)
             edge [<-, out=180, in=-90]                            (5)
         (1) edge [-]                   node[pos=\endPos]  (e2) {} (6)
         (2) edge [-]                   node[pos=\endPos]  (e3) {} (7);
   \begin{pgfonlayer}{background}
    \node[draw=\borderColor,thick,rounded corners,fit = (1) (2) (3) (4) (5) (e1) (e2) (e3)] {};
   \end{pgfonlayer}
  \end{tikzpicture}} \label{gadget:finish:14:1001010}
 \gadgetSubfloat[Gadget \ref{gadget:finish:13:1000000}]{
  \begin{tikzpicture}[scale=\scale,transform shape,>=\arrowType,node distance=\nodeDist,semithick]
   \node[external] (0)              {};
   \node[internal] (5) [right of=0] {};
   \node[internal] (1) [above right of=5] {};
   \node[internal] (2) [below right of=5] {};
   \node[internal] (4) [above of=1] {};
   \node[internal] (3) [below of=2] {};
   \node[external] (6) [right of=1] {};
   \node[external] (7) [right of=2] {};
   \path (0) edge [->]             node[pos=\startPos] (e1) {} (5)
         (1) edge [<-]                                         (4)
             edge [<-, bend right]                             (5)
         (2) edge [<-]                                         (3)
             edge [<-, bend left]                              (5)
         (3) edge [<-, loop below] coordinate          (c1)    (3)
         (4) edge [<-, loop above] coordinate          (c2)    (4)
         (1) edge [-]              node[pos=\endPos]   (e2) {} (6)
         (2) edge [-]              node[pos=\endPos]   (e3) {} (7);
   \begin{pgfonlayer}{background}
    \node[draw=\borderColor,thick,rounded corners,fit = (1) (2) (3) (4) (5) (e1) (e2) (e3) (c1) (c2)] {};
   \end{pgfonlayer}
  \end{tikzpicture}} \label{gadget:finish:13:1000000}
 \gadgetSubfloat[Gadget \ref{gadget:finish:13:1000100}]{
  \begin{tikzpicture}[scale=\scale,transform shape,>=\arrowType,node distance=\nodeDist,semithick]
   \node[external] (0)              {};
   \node[internal] (5) [right of=0] {};
   \node[internal] (1) [above right of=5] {};
   \node[internal] (2) [below right of=5] {};
   \node[internal] (4) [above of=1] {};
   \node[internal] (3) [below of=2] {};
   \node[external] (6) [right of=1] {};
   \node[external] (7) [right of=2] {};
   \path (0) edge [->]             node[pos=\startPos] (e1) {} (5)
         (1) edge [<-]                                         (4)
             edge [<-, bend right]                             (5)
         (2) edge [<-]                                         (3)
             edge [->, bend left]                              (5)
         (3) edge [<-, loop below] coordinate          (c1)    (3)
         (4) edge [<-, loop above] coordinate          (c2)    (4)
         (1) edge [-]              node[pos=\endPos]   (e2) {} (6)
         (2) edge [-]              node[pos=\endPos]   (e3) {} (7);
   \begin{pgfonlayer}{background}
    \node[draw=\borderColor,thick,rounded corners,fit = (1) (2) (3) (4) (5) (e1) (e2) (e3) (c1) (c2)] {};
   \end{pgfonlayer}
  \end{tikzpicture}} \label{gadget:finish:13:1000100}
 \gadgetSubfloat[Gadget \ref{gadget:finish:13:1010100}]{
  \begin{tikzpicture}[scale=\scale,transform shape,>=\arrowType,node distance=\nodeDist,semithick]
   \node[external] (0)              {};
   \node[internal] (5) [right of=0] {};
   \node[internal] (1) [above right of=5] {};
   \node[internal] (2) [below right of=5] {};
   \node[internal] (4) [above of=1] {};
   \node[internal] (3) [below of=2] {};
   \node[external] (6) [right of=1] {};
   \node[external] (7) [right of=2] {};
   \path (0) edge [->]             node[pos=\startPos] (e1) {} (5)
         (1) edge [<-]                                         (4)
             edge [->, bend right]                             (5)
         (2) edge [<-]                                         (3)
             edge [->, bend left]                              (5)
         (3) edge [<-, loop below] coordinate          (c1)    (3)
         (4) edge [<-, loop above] coordinate          (c2)    (4)
         (1) edge [-]              node[pos=\endPos]   (e2) {} (6)
         (2) edge [-]              node[pos=\endPos]   (e3) {} (7);
   \begin{pgfonlayer}{background}
    \node[draw=\borderColor,thick,rounded corners,fit = (1) (2) (3) (4) (5) (e1) (e2) (e3) (c1) (c2)] {};
   \end{pgfonlayer}
  \end{tikzpicture}} \label{gadget:finish:13:1010100}\\
 \gadgetSubfloat[Gadget \ref{gadget:finish:19:1000110}]{
  \begin{tikzpicture}[scale=\scale,transform shape,>=\arrowType,node distance=\nodeDist,semithick]
   \node[external] (0)              {};
   \node[internal] (5) [right of=0] {};
   \node[internal] (4) [above right of=5] {};
   \node[internal] (3) [below right of=5] {};
   \node[internal] (1) [right of=4] {};
   \node[internal] (2) [right of=3] {};
   \node[external] (6) [right of=1] {};
   \node[external] (7) [right of=2] {};
   \path (0) edge [->]             node[pos=\startPos] (e1) {} (5)
         (1) edge [<-, bend right] node[midway]        (e4) {} (4)
             edge [<-, bend left]                              (4)
         (2) edge [<-, bend right]                             (3)
             edge [->, bend left]  node[midway]        (e5) {} (3)
         (3) edge [->, bend left]                              (5)
         (4) edge [<-, bend right]                             (5)
         (1) edge [-]              node[pos=\endPos]   (e2) {} (6)
         (2) edge [-]              node[pos=\endPos]   (e3) {} (7);
   \begin{pgfonlayer}{background}
    \node[draw=\borderColor,thick,rounded corners,fit = (1) (2) (3) (4) (5) (e1) (e2) (e3) (e4) (e5)] {};
   \end{pgfonlayer}
  \end{tikzpicture}} \label{gadget:finish:19:1000110}
 \gadgetSubfloat[Gadget \ref{gadget:finish:19:1010100}]{
  \begin{tikzpicture}[scale=\scale,transform shape,>=\arrowType,node distance=\nodeDist,semithick]
   \node[external] (0)              {};
   \node[internal] (5) [right of=0] {};
   \node[internal] (4) [above right of=5] {};
   \node[internal] (3) [below right of=5] {};
   \node[internal] (1) [right of=4] {};
   \node[internal] (2) [right of=3] {};
   \node[external] (6) [right of=1] {};
   \node[external] (7) [right of=2] {};
   \path (0) edge [->]             node[pos=\startPos] (e1) {} (5)
         (1) edge [<-, bend right] node[midway]        (e4) {} (4)
             edge [->, bend left]                              (4)
         (2) edge [<-, bend right]                             (3)
             edge [->, bend left]  node[midway]        (e5) {} (3)
         (3) edge [<-, bend left]                              (5)
         (4) edge [<-, bend right]                             (5)
         (1) edge [-]              node[pos=\endPos]   (e2) {} (6)
         (2) edge [-]              node[pos=\endPos]   (e3) {} (7);
   \begin{pgfonlayer}{background}
    \node[draw=\borderColor,thick,rounded corners,fit = (1) (2) (3) (4) (5) (e1) (e2) (e3) (e4) (e5)] {};
   \end{pgfonlayer}
  \end{tikzpicture}} \label{gadget:finish:19:1010100}
 \gadgetSubfloat[Gadget \ref{gadget:finish:19:1010111}]{
  \begin{tikzpicture}[scale=\scale,transform shape,>=\arrowType,node distance=\nodeDist,semithick]
   \node[external] (0)              {};
   \node[internal] (5) [right of=0] {};
   \node[internal] (4) [above right of=5] {};
   \node[internal] (3) [below right of=5] {};
   \node[internal] (1) [right of=4] {};
   \node[internal] (2) [right of=3] {};
   \node[external] (6) [right of=1] {};
   \node[external] (7) [right of=2] {};
   \path (0) edge [->]             node[pos=\startPos] (e1) {} (5)
         (1) edge [<-, bend right] node[midway]        (e4) {} (4)
             edge [->, bend left]                              (4)
         (2) edge [<-, bend right]                             (3)
             edge [->, bend left]  node[midway]        (e5) {} (3)
         (3) edge [->, bend left]                              (5)
         (4) edge [->, bend right]                             (5)
         (1) edge [-]              node[pos=\endPos]   (e2) {} (6)
         (2) edge [-]              node[pos=\endPos]   (e3) {} (7);
   \begin{pgfonlayer}{background}
    \node[draw=\borderColor,thick,rounded corners,fit = (1) (2) (3) (4) (5) (e1) (e2) (e3) (e4) (e5)] {};
   \end{pgfonlayer}
  \end{tikzpicture}} \label{gadget:finish:19:1010111}
 \gadgetSubfloat[Gadget \ref{gadget:finish:27:1000010}]{
  \begin{tikzpicture}[scale=\scale,transform shape,>=\arrowType,node distance=\nodeDist,semithick]
   \node[external] (0)              {};
   \node[internal] (5) [right of=0] {};
   \node[internal] (4) [above right of=5] {};
   \node[internal] (3) [below right of=5] {};
   \node[internal] (1) [right of=4] {};
   \node[internal] (2) [right of=3] {};
   \node[external] (6) [right of=1] {};
   \node[external] (7) [right of=2] {};
   \path (0) edge [->]             node[pos=\startPos] (e1) {} (5)
         (1) edge [<-]                                         (2)
             edge [<-]                                         (4)
         (2) edge [<-]                                         (3)
         (3) edge [<-]                                         (4)
             edge [->, bend left]                              (5)
         (4) edge [<-, bend right]                             (5)
         (1) edge [-]              node[pos=\endPos]   (e2) {} (6)
         (2) edge [-]              node[pos=\endPos]   (e3) {} (7);
   \begin{pgfonlayer}{background}
    \node[draw=\borderColor,thick,rounded corners,fit = (1) (2) (3) (4) (5) (e1) (e2) (e3)] {};
   \end{pgfonlayer}
  \end{tikzpicture}} \label{gadget:finish:27:1000010}\\
 \gadgetSubfloat[Gadget \ref{gadget:finish:27:1010011}]{
  \begin{tikzpicture}[scale=\scale,transform shape,>=\arrowType,node distance=\nodeDist,semithick]
   \node[external] (0)              {};
   \node[internal] (5) [right of=0] {};
   \node[internal] (4) [above right of=5] {};
   \node[internal] (3) [below right of=5] {};
   \node[internal] (1) [right of=4] {};
   \node[internal] (2) [right of=3] {};
   \node[external] (6) [right of=1] {};
   \node[external] (7) [right of=2] {};
   \path (0) edge [->]             node[pos=\startPos] (e1) {} (5)
         (1) edge [<-]                                         (2)
             edge [->]                                         (4)
         (2) edge [<-]                                         (3)
         (3) edge [<-]                                         (4)
             edge [->, bend left]                              (5)
         (4) edge [->, bend right]                             (5)
         (1) edge [-]              node[pos=\endPos]   (e2) {} (6)
         (2) edge [-]              node[pos=\endPos]   (e3) {} (7);
   \begin{pgfonlayer}{background}
    \node[draw=\borderColor,thick,rounded corners,fit = (1) (2) (3) (4) (5) (e1) (e2) (e3)] {};
   \end{pgfonlayer}
  \end{tikzpicture}} \label{gadget:finish:27:1010011}
 \caption{Projector gadgets from arity 2 to 1}
\end{figure}
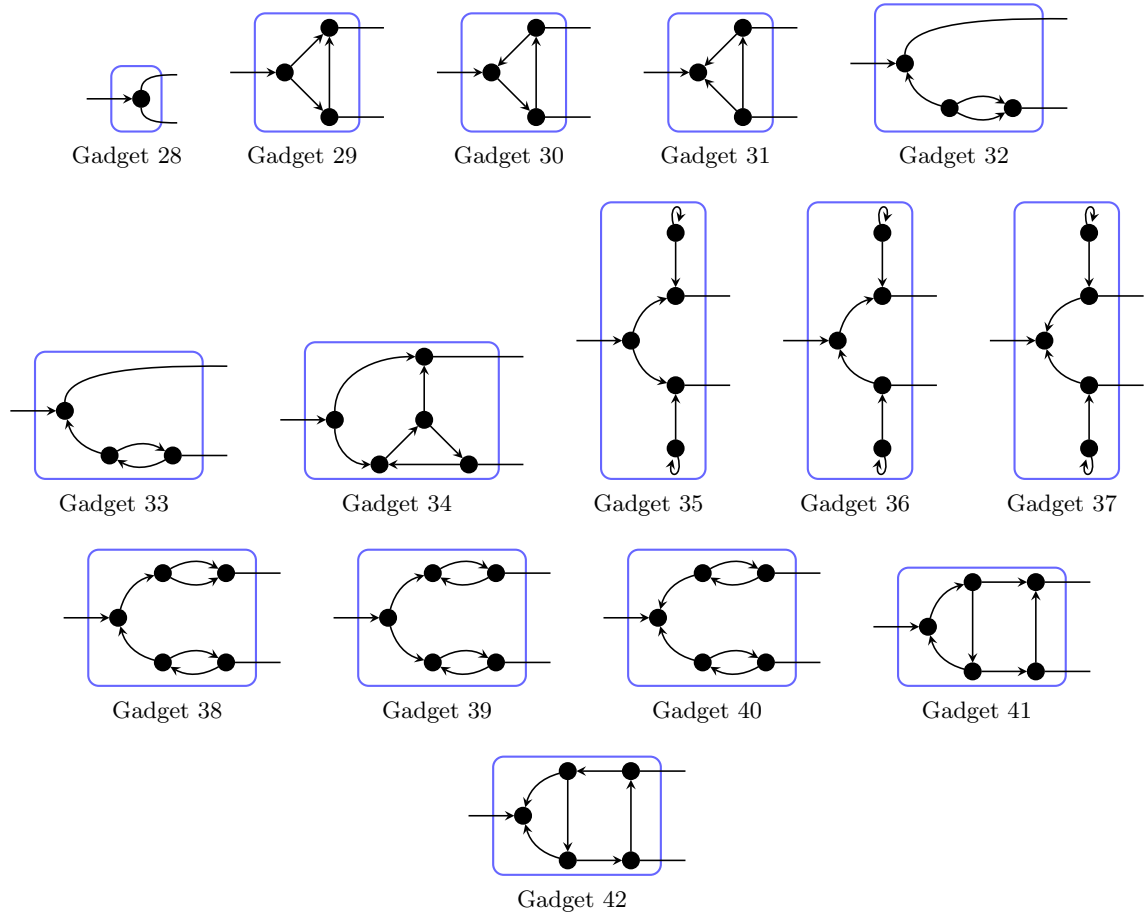
\end{document}